\documentclass[10pt]{bmc_article}    
\usepackage[utf8]{inputenc}
\usepackage[T1]{fontenc}        

\usepackage{verbatim}
\usepackage{cite}
\usepackage[usenames]{color}

\usepackage{amssymb,amsthm,amsfonts,amsmath,enumerate,latexsym,yhmath}
\usepackage{url}
\usepackage{graphicx}
\graphicspath{{images}}
\usepackage{ifthen}  
\usepackage{multicol}   

\newboolean{publ}

\usepackage{tikz}
\usetikzlibrary{snakes,shapes}
\tikzstyle{hyb}=[rectangle,draw,minimum size=3mm]
\tikzstyle{tre}=[circle,draw,minimum size=3.5mm]
\newcommand{\etq}[1]{%
\draw (#1) node {\footnotesize $#1$};
}

\newtheorem{theorem}{Theorem}

\newtheorem{lemma}{Lemma}
\newtheorem{example}{Example}
\newtheorem{corollary}{Corollary}

\newtheorem{proposition}{Proposition}

\renewcommand{\leq}{\leqslant}
\renewcommand{\geq}{\geqslant}
\newcommand{\NN}{\mathbb{N}}

\newcommand{\TT}{\mathcal{T}}
\newcommand{\WTT}{\mathcal{WT}}
\newcommand{\UTT}{\mathcal{UT}}
\newcommand{\BTT}{\mathcal{BT}}
\newcommand{\BT}{\mathcal{BT}}

\newcommand{\RRp}{\mathbb{R}_{>0}}
\newcommand{\RR}{\mathbb{R}}

\def\fT{\varphi_T}
\def\fTp{\varphi_{T'}}

\newenvironment{bmcformat}{\baselineskip20pt\sloppy\setboolean{publ}{false}}{\baselineskip20pt\sloppy}

\begin{document}
\begin{bmcformat}

\title{Cophenetic metrics for phylogenetic trees, after Sokal and Rohlf}

\author{Gabriel Cardona, \email{Gabriel Cardona - gabriel.cardona@uib.es}
Arnau Mir, \email{Arnau Mir - arnau.mir@uib.es}
Francesc Rossell\'o\correspondingauthor,   \email{Francesc Rossell\'o - cesc.rossello@uib.es}
Luc\'\i a Rotger, \email{Luc\'\i a Rotger - lucia.rotger@uib.es}
 and
David S\'anchez \email{David S\'anchez - dscharles@gmail.com}%
}

\address{Department of Mathematics and Computer Science, University of the Balearic Islands, E-07122 Palma de Mallorca, Spain} 

\maketitle

\begin{abstract}
\paragraph*{Background:}  Phylogenetic tree comparison metrics are an important tool in the study of evolution, and hence the definition of such metrics is an interesting problem in phylogenetics.  In a paper in Taxon fifty years ago, Sokal and Rohlf proposed to measure quantitatively the difference between a pair of phylogenetic trees by first encoding them by means of their half-matrices of cophenetic values, and then comparing these matrices. This idea has been used several times since then to define dissimilarity measures between phylogenetic trees but, to our knowledge, no proper metric on weighted phylogenetic trees with nested taxa 
based on this idea has been formally  defined and studied yet. Actually, the cophenetic values of pairs of different taxa alone are not enough to single out phylogenetic trees with weighted arcs or nested taxa.
\paragraph*{Results:}  For every (rooted) phylogenetic tree $T$, let its \emph{cophenetic vector} $\varphi(T)$ consist of all pairs of cophenetic values between pairs of taxa in $T$ and all depths of taxa in $T$. It turns out that these cophenetic vectors  single out weighted phylogenetic trees with nested taxa. We then define a family of cophenetic metrics $d_{\varphi,p}$ by comparing these cophenetic vectors by means of $L^p$ norms, and we study, either analytically or numerically, some of their basic properties: neighbors, diameter, distribution, and their rank correlation with each other and with other metrics. 
\paragraph*{Conclusions:}  The cophenetic metrics can be safely used on weighted phylogenetic trees with nested taxa and no restriction on degrees, and they can be computed in $O(n^2)$ time, where $n$ stands for the number of taxa. The metrics $d_{\varphi,1}$ and $d_{\varphi,2}$ have positive skewed distributions, and they show a low rank correlation with the Robinson-Foulds metric and the nodal metrics, and a very high correlation with each other and with the splitted nodal metrics. The diameter of $d_{\varphi,p}$, for $p\geq 1$, is in $O(n^{(p+2)/p})$, and thus for low $p$ they are more discriminative, having a wider range of values.

\end{abstract}

\ifthenelse{\boolean{publ}}{\begin{multicols}{2}}{}

\section*{Background}
Many phylogenetic trees published in the literature or included in phylogenetic databases  are actually  alternative phylogenies for the same sets of organisms, obtained from different datasets or using different evolutionary models or different phylogenetic reconstruction algorithms  \cite{hoef:05}.   This variety of phylogenetic trees makes it necessary to develop methods for measuring their differences  \cite[Chapter~30]{fel:04}.
 The comparison of phylogenetic trees is also used to compare phylogenetic trees obtained through numerical algorithms with other types of hierarchical classifications \cite{RS:81,sokal.roth:62}, to assess the stability of reconstruction methods \cite{willcliff:taxon71}, 
and in the comparative analysis of dendrograms and other hierarchical cluster structures~\cite{handl.ea 2005,restrepo.ea:2007}. Hence, and since the safest way to quantify the  differences between a pair of trees
is through a metric,  ``tree comparison metrics are an important tool in the study of evolution'' \cite{steelpenny:sb93}.

Many metrics for the comparison of phylogenetic trees  have been proposed so far \cite[Chapter~30]{fel:04}. Some of these metrics are edit distances that count how many  operations of a given type are necessary to transform one tree into the other. These metrics include the nearest-neighbor interchange metric~\cite{waterman.smith:1978} and the subtree prune-and-regrafting distance~\cite{allen.steel:2001}. Other metrics compare a pair of phylogenetic trees through some consensus subtree. This is the case for instance of the MAST distances defined in \cite{MAST1,MAST2,MAST3}. Finally, many metrics for phylogenetic trees are based on the comparison of encodings of the phylogenetic trees, like for instance  the Robinson-Foulds metric~\cite{robinson.foulds:1979,robinson.foulds:mb81} (which can also be understood as an edit distance),  the triples metric \cite{critchlow.ea:1996}, the classical nodal metrics for binary phylogenetic trees \cite{farris:sz69,farris:sz73,phipps:sz71,steelpenny:sb93,willcliff:taxon71}, and the splitted nodal metrics for arbitrary phylogenetic trees \cite{cardona.ea:08a}. The advantage of this last kind of metrics is that, unlike the edit and the consensus distances, they are usually computed in low polynomial time.

In an already fifty years old paper \cite{sokal.roth:62}, Sokal and Rohlf proposed a technique to compare dendrograms (which, in their paper, were equivalent to weighted phylogenetic trees without nested taxa)
on the same set of taxa, by encoding them by means  of  their half-matrices of cophenetic values, and then comparing these structures. Their method runs as follows. To begin with, they divide the range of depths of internal nodes in the tree into a suitable number of equal intervals and number increasingly these intervals. Then, for each pair of taxa $i,j$ in the tree, they compute their \emph{cophenetic value} as the class mark of the interval where the depth of their lowest common ancestor lies. Then, to compare two phylogenetic trees, they compare their corresponding half-matrices of cophenetic values. In that paper, they do it specifically by calculating a correlation coefficient between their entries. Sokal and Rohlf's paper \cite{sokal.roth:62} is quite cited (612 cites according to Google Scholar on July 1, 2012) and their method has been often used to compare hierarchical classifications (see, for instance, \cite{appl1,appl2,appl3}). 

Since Sokal and Rohlf's paper, other papers have compared the half-matrices of cophenetic values to define dissimilarity measures between phylogenetic trees (see, for instance, \cite{Hart67,RS:81}), and such half-matrices  have also been used in the so-called ``comparative method'', the statistical methods used to make inferences on the evolution of a trait among species from the distribution of other traits: see \cite{HP91,P99} and \cite[Chapter~25]{fel:04}.
But, to our knowledge, no proper metric  for phylogenetic trees based on cophenetic values has been formally  defined and studied in the literature. In this paper we define a new family of metrics for weighted phylogenetic trees with nested taxa based on Sokal and Rohlf's idea and we study some of their basic properties: neighbors, diameter, distribution, and their rank correlation with each other and with other metrics. 

Our approach differs in some minors points with Sokal and Rohlf's. For instance, 
we use as the cophenetic value $\varphi(i,j)$ of a pair of taxa $i,j$ the  actual depth of the lowest common ancestor of $i$ and $j$, instead of class marks, which was done by Sokal and Rohlf because of practical limitations. Moreover, instead of using a correlation coefficient, we define metrics by using  $L^p$ norms. Finally, we do not restrict ourselves to dendrograms, without internal labeled nodes, but we also allow nested taxa. 

There is, however, a main difference between our approach and  Sokal and Rohlf's. We do not only consider the cophenetic values of pairs of taxa, but also the depths of the taxa. We must do so because we want to define a metric, where zero distance  means isomorphism, and  the cophenetic values of pairs of different taxa alone do not single out even the dendrograms considered by Sokal and Rohlf. That is, two non isomorphic weighted phylogenetic trees without nested taxa on the same set of taxa can have the same vectors of cophenetic values; see Fig.~\ref{fig:nonuniq}. 

It turns out that the \emph{cophenetic vector} consisting of all cophenetic values of pairs of taxa and the depths of all taxa characterizes a weighted phylogenetic tree with nested taxa. This fact comes from the well known relationship between cophenetic values and patristic distances.  If we denote by $\delta(i)$ the depth of a taxon $i$, by $\varphi(i,j)$ the cophenetic value of a pair of taxa $i,j$ and by $d(i,j)$ the distance  between $i$ and $j$, then  \cite{FKE}
$$
d(i,j)=\delta(i)+\delta(j)-2\varphi(i,j).
$$
So, if the depths of the taxa are known, the knowledge of the cophenetic values of pairs of taxa is equivalent to the knowledge of the additive distance defined by the tree. On their turn, the depths and the additive distance single out the unrooted semi-labelled weighted tree associated to the phylogenetic tree with the former root labeled with a specific label ``root'', and hence the phylogenetic tree itself: cf. Theorem \ref{th:char}.  

The fact that cophenetic vectors single out weighted phylogenetic trees with nested taxa can also be deduced from their relationship with splitted path lengths \cite{cardona.ea:08a}. Recall that the splitted path length $\ell(i,j)$ is the distance from the lowest common ancestor of $i$ and $j$ to $i$. It is known \cite[Thm. 10]{cardona.ea:08a} that the matrix $\big(\ell(i,j)\big)_{i,j}$ characterizes a   weighted phylogenetic tree with nested taxa. Since, obviously,
$$
\ell(i,j)=\delta(i)-\varphi(i,j),
$$
the cophenetic vector uniquely determines the matrix of splitted path lengths, and hence the tree.\footnote{There are some details to be filled here, because for technical reasons we shall allow the root of our phylogenetic trees to have out-degree 1 without being labeled, and this case is not covered by  \cite[Thm. 10]{cardona.ea:08a}, but it is not difficult to modify the argument given above to cover also this case.}

The vector of cophenetic values of pairs of different taxa is also related to the notion of ultrametric \cite{John67,SS:73}. Indeed,  notice that $-\varphi$ satisfies the three-point condition of ultrametrics: for every taxa $i,j,k$,
$$
-\varphi(i,j)\leq \mathrm{min}\{-\varphi(i,k),-\varphi(j,k)\}.
$$
But $-\varphi$ is not an ultrametric, as $\varphi(i,i)=\delta(i)\neq 0$. Actually, $\varphi$ can only be used to define an ultrametric precisely on ultrametric trees, where the depths of all leaves are the same, say $\Delta$. In this case, $\Delta-\varphi$ is the ultrametric defined by the tree. In particular, ultrametric trees can be compared by comparing their vectors of cophenetic values of pairs of different taxa. A similar idea is used in \cite{XAF} to induce an average genetic distance between populations from the average coancestry coefficient.

We would like to dedicate this paper to the memory of Robert R. Sokal, father of the field of numerical taxonomy and who passed away last April. His ideas  permeate biostatistics and computational phylogenetics.

\subsection*{Notations}

A \emph{rooted tree} is a directed finite graph that contains a distinguished node, called the \emph{root}, from which every node can be reached through exactly one path.  A \emph{weighted rooted tree} is a pair $(T, \omega)$ consisting of a rooted tree $T=(V,E)$ and a \emph{weight function} $\omega: E\to \RRp$ that associates to every arc $e\in E$ a non-negative real number $\omega(e)>0$.  We identify every \emph{unweighted} (that is, where no weight function has been explicitly defined) rooted tree $T$ with the weighted rooted tree $(T,\omega)$ with $\omega$ the weight 1 constant function.

Let $T=(V,E)$ be a rooted tree.  Whenever $(u,v)\in E$, we say that $v$ is a \emph{child} of $u$ and that $u$ is the \emph{parent} of $v$. Two nodes with the same parent are \emph{siblings}. The nodes without children are the \emph{leaves} of the tree, and the other nodes (including the root) are called \emph{internal}. A \emph{pendant arc} is an arc ending in a leaf.  The nodes with exactly one child are called \emph{elementary}. A tree is \emph{binary}, or \emph{fully resolved}, when every internal node has exactly two children.

Whenever there exists  a path from a node $u$ to a node $v$, we shall say that $v$ is a  \emph{descendant} of $u$ and also that $u$ is an  \emph{ancestor} of $v$, and we shall denote it by $v\preceq u$; if, moreover, $u\neq v$, we shall write $v\prec u$.  The \emph{lowest common ancestor} (LCA) of a pair of nodes $u,v$ of a rooted tree $T$, in symbols $[u,v]_T$, is the unique common ancestor of them that is a descendant of every other common ancestor of them.    Given a node $v$ of a rooted tree $T$, the \emph{subtree of $T$ rooted at $v$} is the subgraph of $T$ induced on the set of descendants of $v$ (including $v$ itself).  A rooted subtree is a \emph{cherry} when it has 2 leaves,   
a \emph{triplet}, when it has 3 leaves, and   a \emph{quartet},  when it has 4 leaves.

The \emph{distance} from a node $u$ to a descendant $v$ of it  in a weighted rooted tree $T$ is the sum of the
weights of the arcs in the unique  path from $u$ to $v$.  In an unweighted rooted tree, this distance is simply the number of arcs in this path. The \emph{depth}  of a node $v$,  in symbols $\delta_T(v)$, is the distance from the root to $v$.

Let $S$ be a non-empty finite set of  \emph{labels}, or \emph{taxa}.  A (\emph{weighted}) \emph{phylogenetic tree} on $S$ is a (weighted) rooted tree with some of its nodes bijectively labeled in the set $S$, including all its leaves and all its elementary nodes except possibly the root (which can be elementary but unlabeled).  The reasons why we allow  unlabeled elementary roots are that our results are still valid for phylogenetic trees containing them,  and  that even if we forbid them, we would need  in some proofs to use that  Theorem \ref{th:char} below is true for phylogenetic trees containing them. Moreover, it is not uncommon to add an unlabeled elementary root to a phylogenetic tree in some contexts: see, for instance,  the phylogenetic trees depicted in Wikipedia's entry  ``Phylogenetic tree''.\footnote{\url{http://en.wikipedia.org/wiki/Phylogenetic_tree}}

In  a phylogenetic  tree, we shall always identify  a labeled node with its taxon. The internal labeled nodes of a phylogenetic tree are called \emph{nested taxa}. Notice in particular that a phylogenetic tree without nested taxa cannot have elementary nodes other than the root. Although in practice $S$ may be any set of taxa, to fix ideas  we shall usually take $S=\{1,\ldots,n\}$, with $n$ the number of labeled nodes of the tree, and we shall use the term \emph{phylogenetic tree with $n$ taxa} to refer to a phylogenetic tree on this set. In general, we shall denote by $L(T)$ the set of taxa of a phylogenetic tree $T$.

Given a set $S$ of taxa, we shall consider the following spaces of phylogenetic trees:
\begin{itemize}
\item $\WTT(S)$, of all weighted phylogenetic trees on $S$
\item $\UTT(S)$, of all unweighted phylogenetic trees on $S$
\item $\TT(S)$, of all unweighted phylogenetic trees on $S$ without nested taxa
\item $\BT(S)$, of all binary unweighted phylogenetic trees  on $S$ without nested taxa
\end{itemize}
When $S=\{1,\ldots,n\}$, we shall simply write $\WTT_n$, $\UTT_n$,  $\TT_n$, and  $\BT_n$, respectively.

Two phylogenetic trees $T$ and $T'$ on the same set $S$ of taxa are \emph{isomorphic} when they are isomorphic as directed graphs and the isomorphism sends each labeled node of $T$ to the labeled node with the same label in $T'$. An isomorphism of weighted phylogenetic trees is also required to preserve arc weights.  We shall make the abuse of notation of saying that two isomorphic trees are actually the same, and hence of denoting that two trees $T,T'$ are isomorphic by simply writing $T=T'$.

\section*{Methods}

\subsection*{Cophenetic vectors}

Let $S$ be henceforth a non-empty set of taxa with $|S|=n$, which without any loss of generality we identify with $\{1,\ldots,n\}$. Let $T\in \WTT_n$ be a weighted phylogenetic tree on $S$.  For every pair of different taxa $i,j$ in $T$, their \emph{cophenetic value}   is the depth of their LCA:
$$
\varphi_T(i,j)=\delta_T([i,j]_T).
$$ 
To simplify the notations, we shall often write $\varphi_T(i,i)$ to denote the depth $\delta_T(i)$ of a taxon $i$. 

The \emph{cophenetic vector} of $T$ is
$$
\varphi(T)=\big(\varphi_T(i,j)\big)_{1\leq i\leq j\leq n}\in \RR^{n(n+1)/2},
$$
with its elements lexicographically ordered in $(i,j)$. 

\begin{example}
If $T$ is the unweighted phylogenetic tree  in Fig. \ref{fig:ex1}, then $\varphi(T)$ is the vector obtained by
alphabetically ordering in $(i,j)$ the elements of Table \ref{table:ex}.
\end{example}

\begin{center}
\begin{figure}[htb]
\begin{center}
\begin{tikzpicture}[thick,>=stealth,scale=0.3]
\draw(0,0) node [tre] (1) {};  \etq 1
\draw(4,0) node [tre] (2) {};  \etq 2
\draw(6,0) node [tre] (3) {};  \etq 3
\draw(8,0) node [tre] (4) {};  \etq 4  
\draw(10,0) node [tre] (5) {};  \etq 5  
\draw(12,0) node [tre] (6) {};  \etq 6
\draw(1.5,1.5) node [tre] (7) {};  \etq 7
\draw(3,3) node [tre] (a) {};   
\draw(5,5) node [tre] (b) {};   
\draw(7,3) node [tre] (c) {};   
\draw(7,7) node [tre] (r) {};   
\draw(11,3) node [tre] (d) {};   
\draw (r)--(b);
\draw (r)--(d);
\draw (d)--(5);
\draw (d)--(6);
\draw (b)--(c);
\draw (b)--(a);
\draw (c)--(3);
\draw (c)--(4);
\draw (a)--(7);
\draw (a)--(2);
\draw (7)--(1);
\draw(6,-1.5) node {$T$};
\end{tikzpicture}
\end{center}
\caption{\label{fig:ex1} 
An unweighted phylogenetic tree on 7 taxa.}
\end{figure}
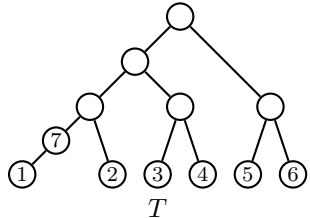
\end{center}
\begin{center}
\begin{table}[htb]
\begin{center}
\begin{tabular}{l|ccccccc}
  ${}_{\textstyle i} \backslash^{\textstyle j}\hspace*{-1ex}$\vphantom{$\int$}& 1 & 2 & 3 & 4 & 5 & 6 & 7\\ \hline
 1 & 4 & 2 & 1 & 1 & 0 & 0 & 3\\
 2 & & 3 & 1 & 1 & 0 & 0 & 2\\
 3 &  & & 3 & 2 & 0 & 0 & 1\\
 4 & & &  & 3 & 0 & 0 & 1\\
 5 & & &  & & 2 & 1 & 0\\
 6 & & &  & &    & 2 & 0\\
 7 & & & &  &    &    & 3
 \end{tabular}
\end{center}
 \caption{\label{table:ex} Cophenetic values of the pairs of taxa in the phylogenetic tree $T$  in Fig. \ref{fig:ex1}.}
\end{table}
\end{center}

The cophenetic vectors single out weighted phylogenetic trees with nested taxa.

\begin{theorem}\label{th:char}
For every  $T,T'\in\WTT(S)$, if $\varphi(T)=\varphi(T')$, then $T= T'$.
\end{theorem}

\begin{proof}
Let $r$ be a symbol not belonging to $S$ and let $X=S\cup\{r\}$. Recall that a  \emph{weighted $X$-tree} is an undirected weighted tree $T$ with set of nodes $V$ endowed with a (non necessarily injective) node-labeling mapping $f:X\to V$ such that $f(X)$ contains all  the leaves and all the degree-2 nodes in $T$ \cite{SS}.

For every $T\in\WTT(S)$, let $T^*$ be the weighted $X$-tree obtained by considering $T$ as undirected and  adding to its former root the label $r$. Then, the distance $d_{T^*}$ on $T^*$  between pairs of labels in $X$ is uniquely determined by $\varphi(T)$ in the following way:
$$
\begin{array}{l}
d_{T^*}(i,r)=\delta_T(i)\qquad \mbox{for every $i\in S$}\\
d_{T^*}(i,j)=\delta_T(i)+\delta_T(j)-2\varphi_T(i,j) \qquad \mbox{for every $i,j\in S$}
\end{array}
$$
Now, $T^*$ is singled out by $d_{T^*}$ \cite[Thm. 7.1.8]{SS}. Since $T$ is uniquely determined from $T^*$ and the knowledge of the root (that is the node labeled with $r$), we deduce that $\varphi(T)$ singles out $T$.
\end{proof}

This result implies that the vectors of cophenetic values of \emph{pairs of different taxa} single out  unweighted phylogenetic trees without nested taxa. 

\begin{corollary}\label{cor:char}
For every $T\in \TT_n$, let 
$\widetilde{\varphi}(T)=\big(\varphi_T(i,j)\big)_{1\leq i< j\leq n}\in \RR^{n(n-1)/2}$,
with its elements lexicographically ordered in $(i,j)$.  Then, for every  $T,T'\in\TT_n$, if $\widetilde{\varphi}(T)=\widetilde{\varphi}(T')$, then $T= T'$. 
\end{corollary}

\begin{proof}
If $T$ is unweighted and without nested taxa, then, for every taxon $i$,
$$
\delta_T(i)=1+\max\{\varphi_T(i,j)\mid 1\leq  j\leq n,\ j\neq i\}
$$
and therefore, in this case, $\varphi(T)$ is uniquely determined by $\widetilde{\varphi}(T)$.
\end{proof} 

But in order to single out phylogenetic trees with non constant weights in the arcs or with nested taxa, it is necessary to take into account also the depths of the leaves. Actually, for example, there is no way to reconstruct from $\widetilde{\varphi}(T)$ the weights of the pendant arcs: the depths of the leaves are needed. Or, without being able to compare depths with cophenetic values, there is no way to say whether a taxon is nested or not. More specifically, for instance, the three trees in Fig.~\ref{fig:nonuniq} have the same value of $\varphi(1,2)$, and hence the same vector  $\widetilde{\varphi}(T)$, but they are not isomorphic as weighted phylogenetic trees.

\begin{figure}[htb]
\begin{center}
\begin{tikzpicture}[thick,>=stealth,scale=0.5]
\draw(0,0) node [tre] (1) {};  \etq 1
\draw(2,0) node [tre] (2) {};  \etq 2
\draw(1,2) node[tre] (r) {};
\draw[->] (r)--node[near end, above] {\footnotesize $4$} (1);
\draw[->] (r)--node[near end, above] {\footnotesize $5$} (2);
\end{tikzpicture}
\quad
\begin{tikzpicture}[thick,>=stealth,scale=0.5]
\draw(0,0) node [tre] (1) {};  \etq 1
\draw(2,0) node [tre] (2) {};  \etq 2
\draw(1,2) node[tre] (r) {};
\draw[->] (r)--node[near end, above] {\footnotesize 1} (1);
\draw[->] (r)--node[near end, above] {\footnotesize 1} (2);
\end{tikzpicture}
\quad
\begin{tikzpicture}[thick,>=stealth,scale=0.5]
\draw(0,0) node [tre] (1) {};  \etq 1
\draw(0,2) node [tre] (2) {};  \etq 2
\draw[->] (2)--node[right] {\footnotesize 1} (1);
\end{tikzpicture}
\end{center}
\caption{\label{fig:nonuniq} 
Three non-isomorphic trees with the same vector $\widetilde{\varphi}(T)$.}
\end{figure}
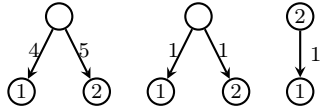

The cophenetic vector $\varphi(T)$ of a weighted phylogenetic tree $T\in\WTT_n$ can be computed in optimal $O(n^2)$ time (assuming a constant cost for the  addition of real numbers) by computing for each internal node $v$,  its depth $\delta_T(v)$ through a preorder traversal of $T$, and the pairs of taxa of which $v$ is the LCA through a postorder traversal of the tree. Both preorder and postorder traversals are performed in linear time on the usual tree data structures.

\subsection*{Cophenetic metrics}

As we have seen in Theorem \ref{th:char},  the mapping
$$
\varphi: \WTT_n \longrightarrow \RR^{n(n+1)/2}
$$
that sends each $T\in \WTT_n$ to its  cophenetic vector $\varphi(T)$, is  injective up to isomorphism. As it is well known, this allows to induce metrics on  $\WTT_n$ from metrics defined on powers of $\RR$. In particular, every $L^p$ norm $\|\, \cdot\, \|_p$ on $\RR^{n(n+1)/2}$, $p\geq 1$, induces a \emph{cophenetic metric} $d_{\varphi,p}$ on $\WTT_n$ by means of
$$
d_{\varphi,p}(T_1,T_2)=\|\varphi(T_1)-\varphi(T_2)\|_p,\quad T_1,T_2\in \WTT_n.
$$
Recall that
$$
\|(x_1,\ldots,x_m)\|_p=\sqrt[p]{|x_1|^p+\cdots+|x_m|^p},
$$
and so, for instance,
$$
\begin{array}{l}
d_{\varphi,1}(T_1,T_2)=\displaystyle \sum_{1\leq i\leq j\leq
n}|\varphi_{T_1}(i,j)-\varphi_{T_2}(i,j)|\\[3ex]
d_{\varphi,2}(T_1,T_2)= \displaystyle\sqrt{\sum_{1\leq i\leq j\leq
n}(\varphi_{T_1}(i,j)-\varphi_{T_2}(i,j))^2}
\end{array}
$$
are the cophenetic metrics on $\WTT_n$ induced by the Manhattan $L^1$ and the euclidean $L^2$ norms.
One can also use  Donoho's $L^0$ ``norm'' (which, actually, is not a proper norm)
$$
\|(x_1,\ldots,x_m)\|_0=\mbox{number of entries $x_i$ that are $\neq 0$}
$$
to induce a metric $d_{\varphi,0}(T_1,T_2)$ on $\WTT_n$, which turns out to be simply the Hamming distance between $\varphi(T_1)$ and $\varphi(T_2)$.

As we have seen in the previous subsection, the cophenetic vector of a phylogenetic tree in $\WTT_n$ can be computed in $O(n^2)$ time. For every $T_1,T_2\in \WTT_n$, and assuming a constant cost for the  addition and product of real numbers, the cost of computing $d_{\varphi,0}(T_1,T_2)$ (as the number of non-zero entries of  $\varphi(T_1)-\varphi(T_2)$) is $O(n^2)$, and the cost of computing $d_{\varphi,p}(T_1,T_2)^p$, for $p\geq 1$ (as the sum of the $p$-th powers of the entries of the difference $\varphi(T_1)-\varphi(T_2)$) is $O(n^2+\log_2(p)n^2)$, which is again $O(n^2)$ if we understand  $\log(p)$ as part of the constant factor.  Finally, the cost of computing $d_{\varphi,p}(T_1,T_2)$,  $p\geq 1$, as the $p$-th root of $d_{\varphi,p}(T_1,T_2)^p$ will depend on $p$ and on the accuracy with which this root is computed. Assuming a constant cost for the computation of $p$-th roots with a given accuracy (notice that, in practice, for low $p$ and accuracy,  this step will be dominated by the computation of 
$d_{\varphi,p}(T_1,T_2)^p$),  the total cost of computing $d_{\varphi,p}(T_1,T_2)$ is  $O(n^2)$.

Next examples show some features of these cophenetic metrics.

\begin{example}\label{ex:contract}
Let $T\in\UTT_n$, let $(u,v)$ be an arc of $T$ with $u$ or $v$ unlabeled, and let $T'$ be the phylogenetic tree in $\UTT_n$ obtained by \emph{contracting} $(u,v)$: that is, by removing the node $v$ and the arc $(u,v)$, labeling $u$ with the label of $v$ if it was labeled, and  replacing every arc $(v,x)$ in $T$ by an arc $(u,x)$.  Notice that, in the passage from $T$ to $T'$, for every $i,j\in S$:
\begin{itemize}
\item If both $i,j$ are descendants of $v$ in $T$, then $\varphi_{T'}(i,j)=\varphi_T(i,j)-1$.

\item In any other case, $\varphi_{T'}(i,j)=\varphi_T(i,j)$.
\end{itemize}
As a consequence, 
$$
\varphi_{T}(i,j)-\varphi_{T'}(i,j)=\left\{
\begin{array}{ll}
1 & \mbox{ if $i,j\preceq v$}\\
0 & \mbox{otherwise}
\end{array}\right.
$$
and therefore,
if  $n_v$ is the number of descendant taxa of $v$,
$$
d_{\varphi,0}(T,T')= {\binom{n_v+1}{2}},\quad d_{\varphi,p}(T,T')=\sqrt[p]{\binom{n_v+1}{2}}\mbox{ if $p\geq 1$}.
$$
So the contraction of an arc  in an tree $T$ (which is Robinson-Foulds' $\alpha$-operation \cite{robinson.foulds:mb81}) yields a new tree $T'$ at a cophenetic distance from $T$ that depends increasingly on the number of descendant taxa of the head of the contracted arc. \qed
 \end{example}
 
 \begin{example}
 Let $T_0,T_0'\in \WTT_m$, for some $m<n$, let $T\in \WTT_n$ be such that its subtree rooted at some node $z$ is $T_0$, and let $T'\in \WTT_n$ be the tree obtained by replacing in $T$ this subtree $T_0$ by $T_0'$.
  
Notice that, for every $i,j\in \{1,\ldots,n\}$,
$\varphi_{T}(i,j)=\delta_T(z)+\varphi_{T_0}(i,j)$ if $i,j\leq m$, and 
$\varphi_{T}(i,j)=\varphi_T(z,j)$ if $i\leq m$ and $j>m$,
and the same holds in $T'$, replacing $T$ and $T_0$ by $T'$ and $T_0'$, respectively. Since, moreover,
$\delta_T(z)=\delta_{T'}(z)$, $\varphi_T(z,j)=\varphi_{T'}(z,j)$ for every $j>m$, and 
$\varphi_T(i,j)=\varphi_{T'}(i,j)$ for every $i,j>m$,
we conclude that
$$
\varphi(T)-\varphi(T')=\varphi(T_0)-\varphi(T_0')
$$
and hence
$$
d_{\varphi,p}(T,T')=d_{\varphi,p}(T_0,T_0').
$$
So, the cophenetic metrics are local, as other popular metrics like 
the Robinson Foulds or the triples metrics, but unlike other popular metrics, like for instance the nodal metrics.\qed
\end{example}

\section*{Results}

\subsection*{Minimum and maximum values for cophenetic metrics}

Our first goal is to find the smallest non-negative value of $d_{\varphi,p}$ on several spaces of phylogenetic trees, and the pairs of trees at which it is reached. These pairs of  trees at minimum distance can be understood as `adjacent' in the corresponding metric space, and their characterization yields a first step towards  understanding how cophenetic metrics measure the difference between two trees.

Notice that this problem makes no sense for weighted phylogenetic trees. For instance, if we add or subtract an $\varepsilon>0$ to the weight of a pendant arc in a tree $T$, without changing its topology,  the distance between $T$ and the resulting tree will be $\varepsilon$, which can be as small as desired. So, we only consider this problem on $\UTT_n$, $\TT_n$, and $\BTT_n$. 

In order to simplify the statements, set
$$
D_p(T_1,T_2)=\left\{\begin{array}{ll}
d_{\varphi,0}(T_1,T_2) & \mbox{ if $p=0$}\\
d_{\varphi,p}(T_1,T_2)^p & \mbox{ if $p\geq 1$}\
\end{array}\right.
$$

The following easy result, which is a direct consequence of the fact that $D_p(T_1,T_2)\geq D_0(T_1,T_2)$ for every $p\geq 1$ and $T_1,T_2\in \UTT_n$, will be used in the proof of the next propositions.

\begin{lemma}\label{lem:pvs0}
Assume that, for every pair of different trees $T_1,T_2$ in~$\UTT_n$, $\TT_n$ or $\BTT_n$ such that $D_0(T_1,T_2)$ is minimum on this space, we have that $D_p(T_1,T_2)= D_0(T_1,T_2)$. Then, the minimum non-zero value of $D_p$ on this space of trees is equal to the minimum non-zero value of $D_0$ on it, and it is reached at exactly the same pairs of trees. \qed
\end{lemma}

The least non-negative values of $D_p$, for $p\in \{0\}\cup[1,\infty[$, on $\UTT_n$, $\TT_n$, and $\BTT_n$, together with an explicit description of  the pairs of trees where these minimum values are reached,  are given by the next three propositions. We give their proofs in the Appendix.

\begin{proposition}\label{prop:minutt}
The minimum non-negative value of $D_p$ on $\UTT_n$, for $p\in \{0\}\cup[1,\infty[$ and $n\geq 2$, is 1. And for every $T,T'\in \UTT_n$, $D_p(T,T')=1$ if, and only if, one of them is obtained from the other by contracting a pendant arc. \qed
\end{proposition}

So, not every tree in $\UTT_n$ has  neighbors at cophenetic distance 1: only those trees with some leaf whose parent is unlabeled. Now, it is not difficult to check that a tree $T\in \UTT_n$ such that all its leaves have labeled parents has some tree $T'$ such that $D_p(T,T')=2$, which is the minimum value of $D_{p}$ on $\UTT_n$ greater than 1. One such $T'$ is obtained by choosing a pendant arc in $T$ and interchanging the labels of its source and its target nodes.

\begin{proposition}\label{prop:mintt}
The minimum non-negative value of $D_p$ on $\TT_n$, for $p\in \{0\}\cup[1,\infty[$ and $n\geq 3$,  is $3$. And for every $T,T'\in \TT_n$, $D_p(T,T')=3$ if, and only if, one of them is obtained from the other by means of one of the following two operations:
\begin{enumerate}[(a)]
\item Contracting an arc ending in the parent of a cherry (see Fig. \ref{fig:prop21})
\item Pruning and regrafting a leaf that is a  sibling of the root of a cherry, to make it a sibling of the leaves in the cherry  (see Fig.  \ref{fig:prop22}) \qed
\end{enumerate}
\end{proposition}

\begin{center}
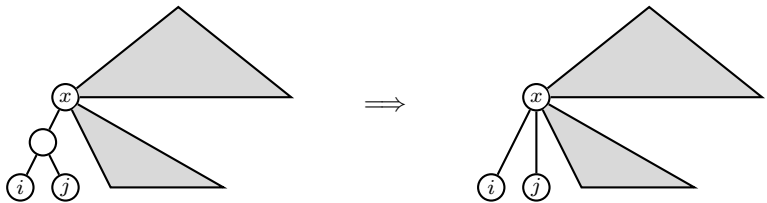
\begin{figure}[htb]
\begin{center}
\begin{tikzpicture}[thick,>=stealth,scale=0.3]
\draw(0,0) node [tre] (i) {};  \etq i
\draw(2,0) node [tre] (j) {};  \etq j 
\draw(1,2) node [tre] (a) {};   
\draw(2,4) node [tre] (x) {};  \etq x 
\draw (x)--(a);
\draw (a)--(j);
\draw (a)--(i);
\draw[fill= black!15] (x)--(4,0)--(9,0)--(x);
\draw[fill= black!15] (x)--(12,4)--(7,8)--(x);
\end{tikzpicture}
\qquad
\begin{tikzpicture}[thick,>=stealth,scale=0.3]
\draw(0,0) node  {};   
\draw(0,8) node  {}; 
\draw(1,4) node   {$\Longrightarrow$};   
\end{tikzpicture}
\qquad
\begin{tikzpicture}[thick,>=stealth,scale=0.3]
\draw(0,0) node [tre] (i) {};  \etq i
\draw(2,0) node [tre] (j) {};  \etq j 
\draw(2,4) node [tre] (x) {};  \etq x 
\draw (x)--(j);
\draw (x)--(i);
\draw[fill= black!15] (x)--(4,0)--(9,0)--(x);
\draw[fill= black!15] (x)--(12,4)--(7,8)--(x);
\end{tikzpicture}
\end{center}
\caption{\label{fig:prop21} Contraction of an arc ending in the parent of a cherry.}
\end{figure}
\end{center}

\begin{center}
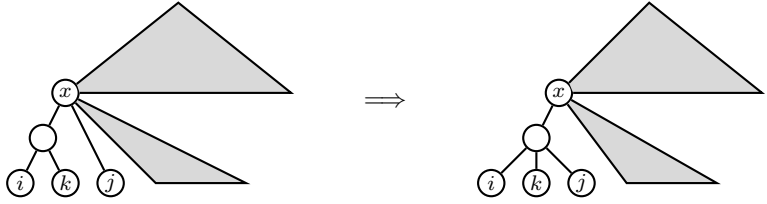
\begin{figure}[htb]
\begin{center}
\begin{tikzpicture}[thick,>=stealth,scale=0.3]
\draw(0,0) node [tre] (i) {};  \etq i
\draw(2,0) node [tre] (k) {};  \etq k
\draw(4,0) node [tre] (j) {};  \etq j 
\draw(1,2) node [tre] (a) {};   
\draw(2,4) node [tre] (x) {};  \etq x 
\draw (x)--(a);
\draw (x)--(j);
\draw (a)--(i);
\draw (a)--(k);
\draw[fill= black!15] (x)--(6,0)--(10,0)--(x);
\draw[fill= black!15] (x)--(12,4)--(7,8)--(x);
\end{tikzpicture}
\qquad
\begin{tikzpicture}[thick,>=stealth,scale=0.3]
\draw(0,0) node  {};   
\draw(0,8) node  {}; 
\draw(1,4) node   {$\Longrightarrow$};   
\end{tikzpicture}
\qquad
\begin{tikzpicture}[thick,>=stealth,scale=0.3]
\draw(0,0) node [tre] (i) {};  \etq i
\draw(2,0) node [tre] (k) {};  \etq k
\draw(4,0) node [tre] (j) {};  \etq j 
\draw(2,2) node [tre] (a) {};   
\draw(3,4) node [tre] (x) {};  \etq x 
\draw (x)--(a);
\draw (a)--(j);
\draw (a)--(i);
\draw (a)--(k);
\draw[fill= black!15] (x)--(6,0)--(10,0)--(x);
\draw[fill= black!15] (x)--(12,4)--(7,8)--(x);
\end{tikzpicture}
\end{center}
\caption{\label{fig:prop22} 
Pruning and regrafting an uncle of a cherry to make it a sibling of them.}
\end{figure}
\end{center}

So, every tree $T\in \TT_n$ has neighbors $T'$ such that $D_p(T,T')=3$. Indeed, take an internal node $v$ in $T$ of largest depth, so that all its children are leaves. If $v$ has exactly two children, one such neighbor of $T$ is obtained by contracting the arc ending in $v$. If $v$ has more than two children, one such neighbor of $T$ is obtained by replacing any two children of $v$ by a cherry (that is, taking two children $i,j$ of $v$, removing the arcs $(v,i)$ and $(v,j)$, and then adding a new node $v_0$ and arcs $(v,v_0)$, $(v_0,i)$, and $(v_0,j)$).

\begin{proposition}\label{prop:minbtt}
The minimum non-negative value of $D_p$ on $\BTT_n$, for $p\in \{0\}\cup[1,\infty[$ and $n\geq 3$,  is $4$. And for every $T,T'\in \BTT_n$, $D_p(T,T')=4$ if, and only if, one of them is obtained from the other by means of one of the following  operations:
\begin{enumerate}[(a)]
\item Reorganizing a triplet (see Fig. \ref{fig:dist4tri-2})
\item Reorganizing a completely branched quartet  (see Fig.  \ref{fig:dist4tfour-2}) \qed
\end{enumerate}
\end{proposition}

\begin{center}
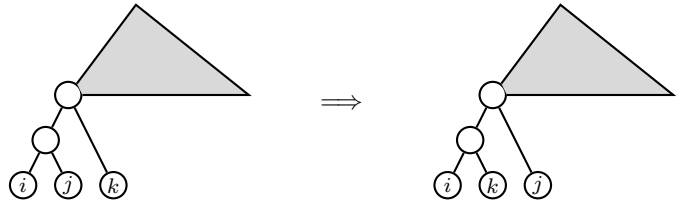
\begin{figure}[htb]
\begin{center}
\begin{tikzpicture}[thick,>=stealth,scale=0.3]
\draw(0,0) node [tre] (i) {};  \etq i
\draw(2,0) node [tre] (j) {};  \etq j
\draw(4,0) node [tre] (k) {};  \etq k
\draw(1,2) node [tre] (a) {};   
\draw(2,4) node [tre] (x) {};  
\draw[fill=black!15] (x)--(10,4)--(5,8)--(x);
\draw (x)--(a);
\draw (a)--(i);
\draw (a)--(j);
\draw (x)--(k);
\end{tikzpicture}
\qquad
\begin{tikzpicture}[thick,>=stealth,scale=0.3]
\draw(0,0) node {};
\draw(0,8) node {};
\draw(0,4) node {$\Longrightarrow$};
\end{tikzpicture}
\qquad
\begin{tikzpicture}[thick,>=stealth,scale=0.3]
\draw(0,0) node [tre] (i) {};  \etq i
\draw(2,0) node [tre] (k) {};  \etq k
\draw(4,0) node [tre] (j) {};  \etq j
\draw(1,2) node [tre] (a) {};   
\draw(2,4) node [tre] (x) {};  
\draw[fill=black!15] (x)--(10,4)--(5,8)--(x);
\draw (x)--(a);
\draw (a)--(i);
\draw (a)--(k);
\draw (x)--(j);
\end{tikzpicture}
\end{center}
\caption{\label{fig:dist4tri-2} 
Reorganizing a triplet.}
\end{figure}
\end{center}

\begin{center}
\begin{figure}[htb]
\begin{center}
\begin{tikzpicture}[thick,>=stealth,scale=0.3]
\draw(0,0) node [tre] (i) {};  \etq i
\draw(2,0) node [tre] (j) {};  \etq j
\draw(4,0) node [tre] (k) {};  \etq k
\draw(6,0) node [tre] (l) {};  \etq l
\draw(1,2) node [tre] (a) {};   
\draw(5,2) node [tre] (b) {};   
\draw(3,4) node [tre] (x) {};  
\draw[fill=black!15] (x)--(12,4)--(7,8)--(x);
\draw (x)--(a);
\draw (x)--(b);
\draw (a)--(i);
\draw (a)--(j);
\draw (b)--(l);
\draw (b)--(k);
\end{tikzpicture}
\qquad
\begin{tikzpicture}[thick,>=stealth,scale=0.3]
\draw(0,0) node {};
\draw(0,8) node {};
\draw(0,4) node {$\Longrightarrow$};
\end{tikzpicture}
\qquad
\begin{tikzpicture}[thick,>=stealth,scale=0.3]
\draw(0,0) node [tre] (i) {};  \etq i
\draw(2,0) node [tre] (k) {};  \etq k
\draw(4,0) node [tre] (j) {};  \etq j
\draw(6,0) node [tre] (l) {};  \etq l
\draw(1,2) node [tre] (a) {};   
\draw(5,2) node [tre] (b) {};   
\draw(3,4) node [tre] (x) {};  
\draw[fill=black!15] (x)--(12,4)--(7,8)--(x);
\draw (x)--(a);
\draw (x)--(b);
\draw (a)--(i);
\draw (a)--(k);
\draw (b)--(l);
\draw (b)--(j);
\end{tikzpicture}
\end{center}
\caption{\label{fig:dist4tfour-2} 
Reorganizing a completely branched quartet.}
\end{figure}
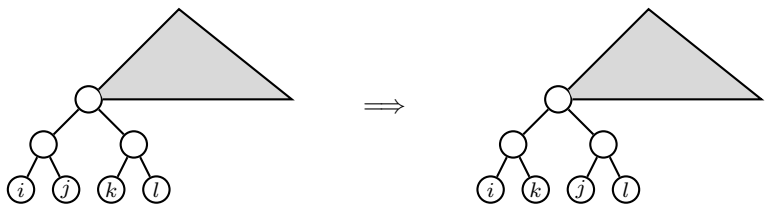
\end{center}

So again, every tree $T\in \BTT_n$ has neighbors $T'$ such that $D_p(T,T')=4$. Indeed, take an internal node $v$ in $T$ of largest depth, so that its two children are leaves. Let $w$ be the parent of $v$. Then, either the other child of $w$ is a leaf, in which case $w$ is the root of a triple  and reorganizing its taxa we obtain a neighbor of $T$, or the other child of $w$ is the parent of a cherry (it will have the same, maximum, depth as $v$), in which case $w$ is the root of a completely branched quartet and reorganizing its taxa we obtain a neighbor of $T$.

We focus now on the diameter, that is, the largest value of $d_{\varphi,p}$ on the spaces of unweighted phylogenetic trees (as in the case of the minimum non-zero value, and for the same reasons,  the  problem of finding the diameter makes no sense for weighted trees). Unfortunately, we have not been able to find  exact formulas for it, but we have obtained its order, which we give in the next proposition. We also give its proof  in the Appendix.

\begin{proposition}\label{prop:diam}
The diameter of $d_{\varphi,p}$  on $\UTT_n$, $\TT_n$, and $\BTT_n$ is in $\Theta(n^2)$ if $p=0$ and in $\Theta(n^{(p+2)/p})$ if $p\geq 1$.
\end{proposition}

In particular, the diameter of $d_{\varphi,1}$ on these spaces  is in $\Theta(n^{3})$, and  the diameter of $d_{\varphi,2}$   is in $\Theta(n^{2})$.

\subsection*{Numerical experiments}

We have performed several numerical experiments concerning the distributions of $d_{\varphi,1}$ and $d_{\varphi,2}$, and the correlation of these metrics with other phylogenetic tree comparison metrics. The results of all these experiments can be found in the Supplementary Material web page \url{http://bioinfo.uib.es/~recerca/phylotrees/cophidist/}. In this section we report only on some  significant results obtained through these experiments.

As a first experiment, we have generated all trees in $\BTT_n$ and $\TT_n$, for $n=3,4,5,6$, and for all pairs of them we have computed:
\begin{itemize}
\item The cophenetic distances $d_{\varphi,1}$ and $d_{\varphi,2}$ on $\BTT_n$ and $\TT_n$.
  
\item The Robinson-Foulds distance $d_{\mathrm{RF}}$  on $\BTT_n$ and $\TT_n$ \cite{robinson.foulds:mb81}.

\item The classical  nodal distances $d_{\mathrm{nodal},1}$  and $d_{\mathrm{nodal},2}$ on $\BTT_n$, which compare the vectors of  distances between pairs of taxa by means of the Manhattan and the Euclidean norms, respectively; see  \cite{willcliff:taxon71} and \cite{farris:sz73}, respectively, as well as \cite{cardona.ea:08a}.

\item The splitted nodal distances $d_{\mathrm{nodal},1}^{\mathrm{sp}}$  and
  $d_{\mathrm{nodal},2}^{\mathrm{sp}}$  on $\TT_n$,  which compare the matrices of splitted path lengths between pairs of taxa by means of the Manhattan and the Euclidean norms, respectively; see \cite{cardona.ea:08a}.
\end{itemize}

In order to analyze this data, we have plotted 2D-histograms for all pairs of metrics and
we have computed their Spearman's rank correlation
coefficient. 
On the one hand, the 2D-histograms for $\BTT_6$ and $\TT_6$ (the most significative case) are given in
Figures~\ref{fig:2dhisto-petit-bin} and \ref{fig:2dhisto-petit-nobin}, respectively. For each pair of distances, we have
divided the range of values that each of the distances 
gets into $25$ subranges, and computed how many pairs of trees fall
into each of the $25\times25$ different possibilities. Each of these
possibilities is represented by a rectangle in a grid, whose
darkness level is proportional of the number of trees.
On the other hand, the Spearman's rank correlation
coefficient between the aforementioned distances in the most
significative case of $n=6$ are given in
Tables~\ref{tab:spearman_petit-bin} and \ref{tab:spearman_petit-nobin}.

 \begin{center}
\begin{figure}[htb]
\begin{center}
\includegraphics[width=0.8\linewidth]{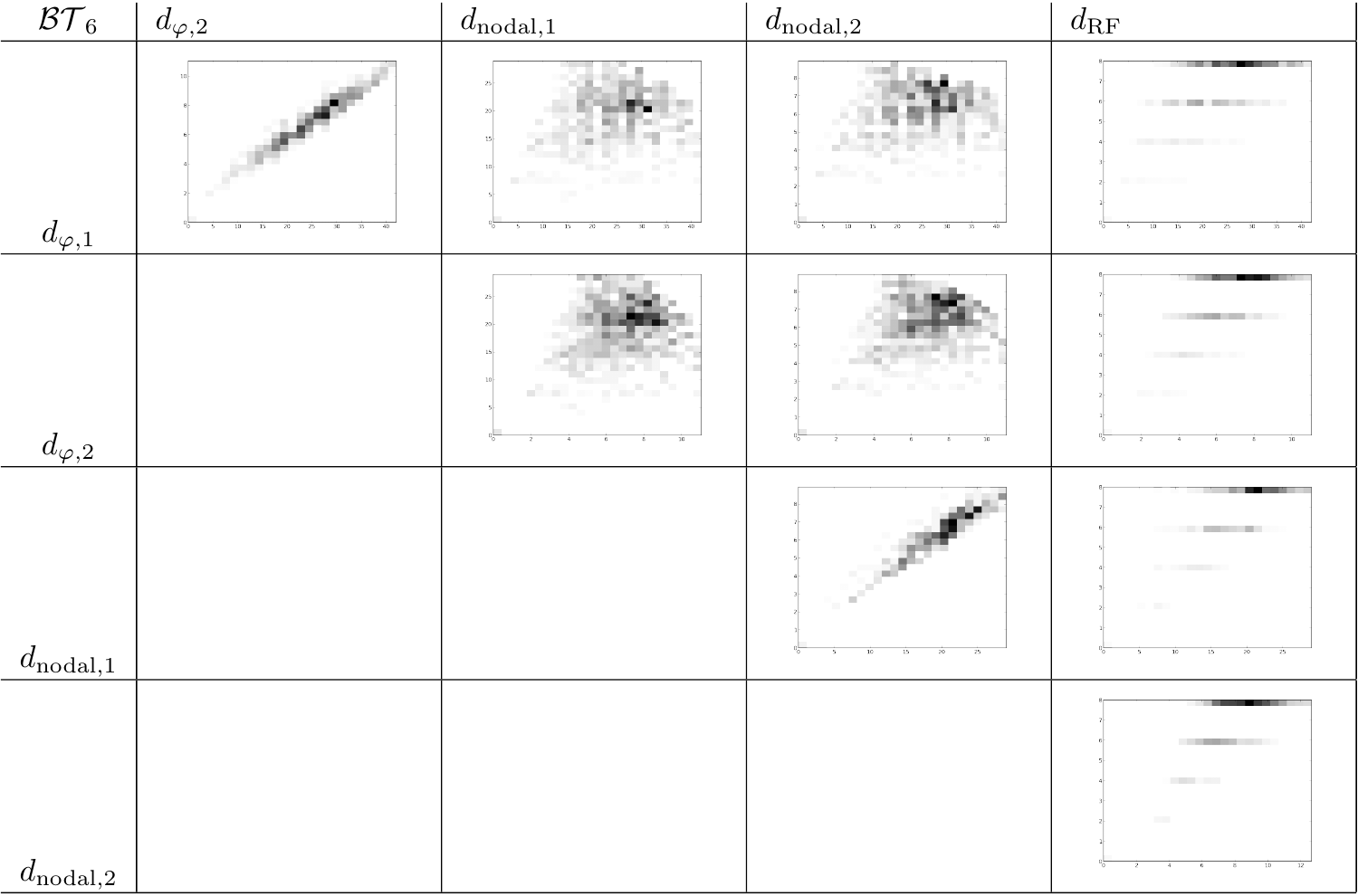}
\end{center}
\caption{\label{fig:2dhisto-petit-bin} 2D-histograms showing the relationship between different distances on $\BTT_6$.}
\end{figure}
\end{center}

 \begin{center}
\begin{figure}[htb]
\begin{center}
\includegraphics[width=0.8\linewidth]{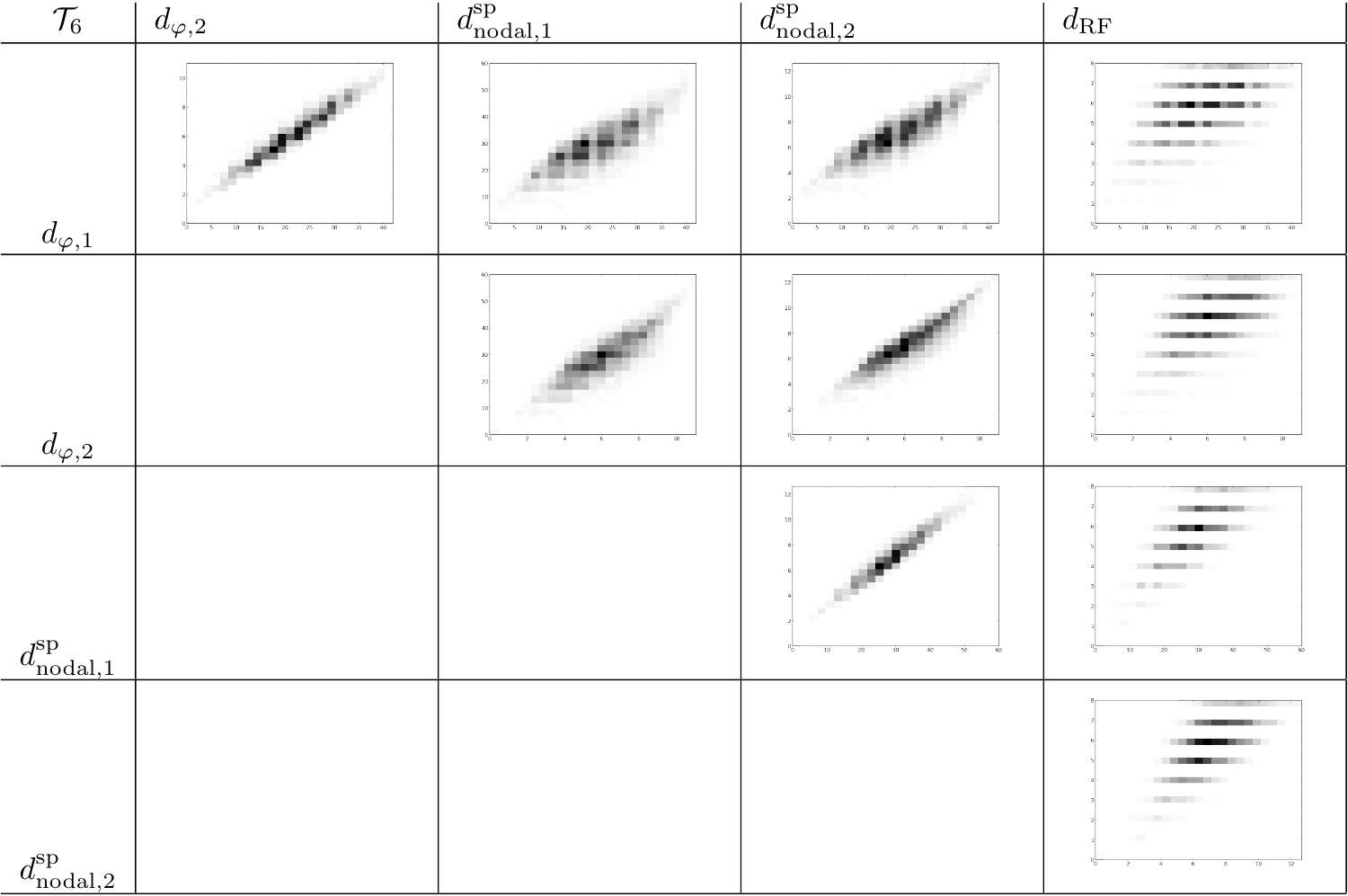}
\end{center}
\caption{\label{fig:2dhisto-petit-nobin} 2D-histograms showing the relationship between different distances on $\TT_6$.}
\end{figure}
\end{center}

 \begin{center}
\begin{table}[htb]
 \begin{center}
     \begin{tabular}{c|l|l|l|l|}
    $\BTT_6$ & $d_{\varphi,2}$ & $d_{\rm{nodal},1}$ & $d_{\rm{nodal},2}$ & $d_{\rm{RF}}$\\\hline
    $d_{\varphi,1}$ & 0.966309 & 0.066217 & 0.057751 & 0.473775\\\hline
    $d_{\varphi,2}$ & &0.093708 & 0.100914& 0.501130\\\hline
    $d_{\rm{nodal},1}$ & & &0.928421 & 0.585127\\\hline
    $d_{\rm{nodal},2}$ & & & & 0.623644\\\hline
 \end{tabular}
\end{center}
 \caption{\label{tab:spearman_petit-bin} Spearman's rank correlation coefficient between different distances on $\BTT_6$.}
\end{table}
\end{center}

\begin{center}
\begin{table}[htb]
 \begin{center}
    \begin{tabular}{c|l|l|l|l|}
    $\TT_6$ & $d_{\varphi,2}$ & $d_{\rm{nodal},1}^{\rm{sp}}$ & $d_{\rm{nodal},2}^{\rm{sp}}$ & $d_{\rm{RF}}$\\\hline
    $d_{\varphi,1}$ & 0.965115 & 0.803159 & 0.864113 & 0.505631 \\\hline
    $d_{\varphi,2}$ & & 0.831387 & 0.902573 & 0.529837\\\hline
    $d_{\rm{nodal},1}^{\rm{sp}}$ & & & 0.957057 & 0.665752 \\\hline
    $d_{\rm{nodal},2}^{\rm{sp}}$ & & & & 0.642203\\\hline
 \end{tabular}
\end{center}
 \caption{\label{tab:spearman_petit-nobin} Spearman's rank correlation coefficient between different
    distances on $\TT_6$.}
\end{table}
\end{center}

These histograms and tables show that $d_{\varphi,1}$ and $d_{\varphi,2}$ are highly correlated, and that each $d_{\varphi,i}$, $i=1,2$, is highly correlated with the corresponding $d_{\mathrm{nodal},i}^{\rm{sp}}$ on $\TT_6$. This is not a surprise, because both types of metrics are based on encodings of phylogenetic trees related to the position in the tree of the LCA of every pair of leaves: remember the relationship between depths, cophenetic values and splitted path lengths recalled in the Background section. More surprising to us is the low correlation between  each $d_{\varphi,i}$,  and the corresponding $d_{\mathrm{nodal},i}$ on $\BTT_6$, because of the  relationship between depths, cophenetic values and patristic distances also recalled in the Background section. The very low correlation between the cophenetic metrics and the Robinson-Foulds metric simply shows that these metrics measure different notions of similarity.

Our second experiment is for values of $n$ greater than $6$.
The numbers of trees in each of the
spaces $\TT_n$ and $\BTT_n$ make it unfeasible to compute the
distances between all pairs of trees. Hence, we have randomly and uniformly
generated pairs of trees in each of these spaces for
$n=10,20,\dots,100$ until the approximated value of the Spearman's
rank correlations of all pairs of distances converge up to 3
significant digits.
The corresponding  2D-histograms and Spearman's rank correlation
coefficient  tables for the most significative case of $n=100$ are shown
in Figures~\ref{fig:2dhisto-bin} and \ref{fig:2dhisto-nobin} and
Tables~\ref{tab:spearman-bin} and \ref{tab:spearman-nobin}. These diagrams and tables confirm the very high correlation between
$d_{\varphi,1}$ and $d_{\varphi,2}$, and very low correlation of these metrics and the nodal and Robinson-Foulds metrics. The correlation between each $d_{\varphi,i}$, $i=1,2$, and  the corresponding $d_{\mathrm{nodal},i}^{\rm{sp}}$ is still significant, but it decreases as $n$ increases.

 \begin{center}
\begin{figure}[htb]
\begin{center}
\includegraphics[width=0.8\linewidth]{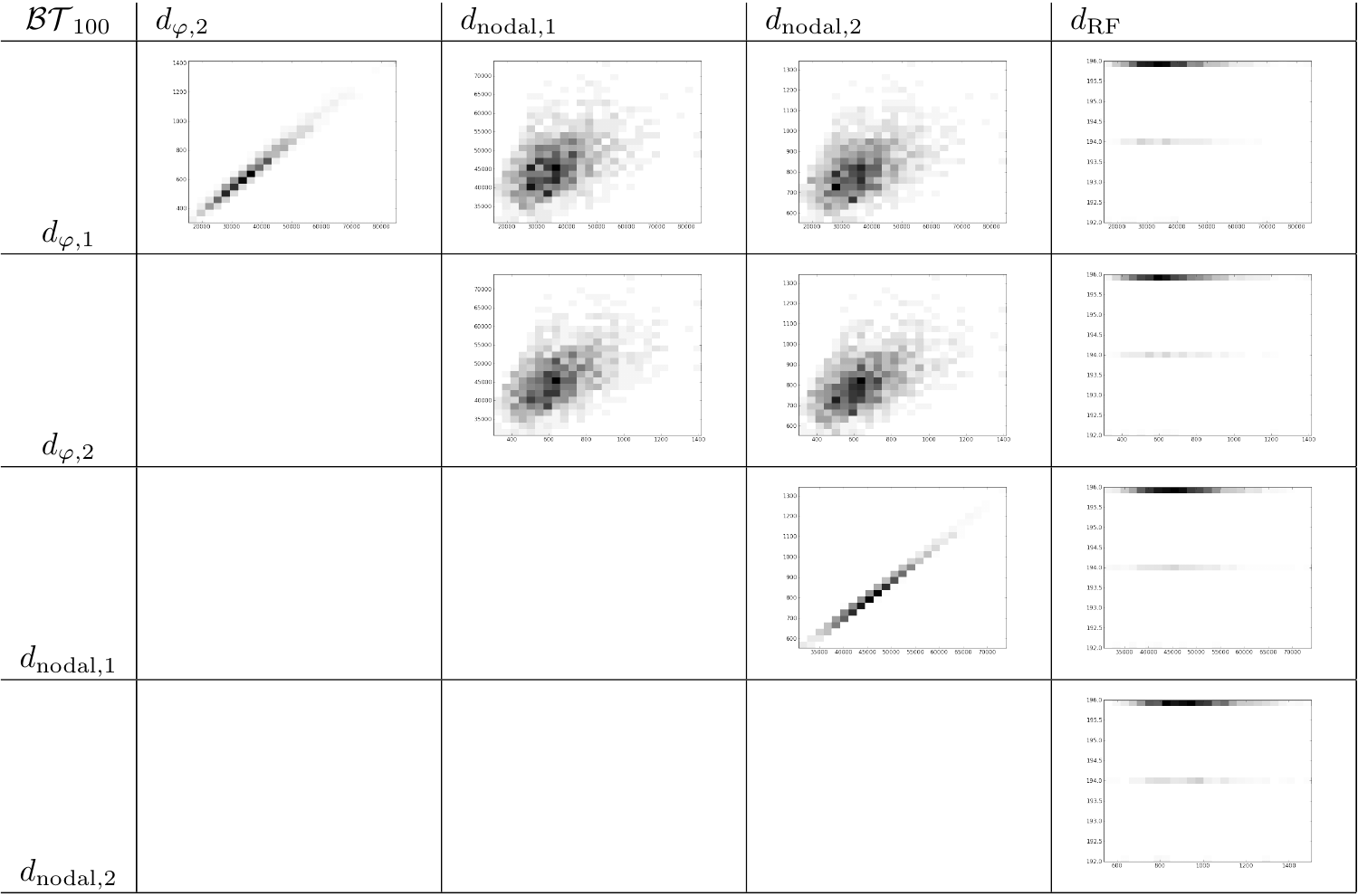}
\end{center}
\caption{\label{fig:2dhisto-bin} 2D-histograms showing the relationship between different distances on $\BTT_6$.}
\end{figure}
\end{center}

 \begin{center}
\begin{figure}[htb]
\begin{center}
\includegraphics[width=0.8\linewidth]{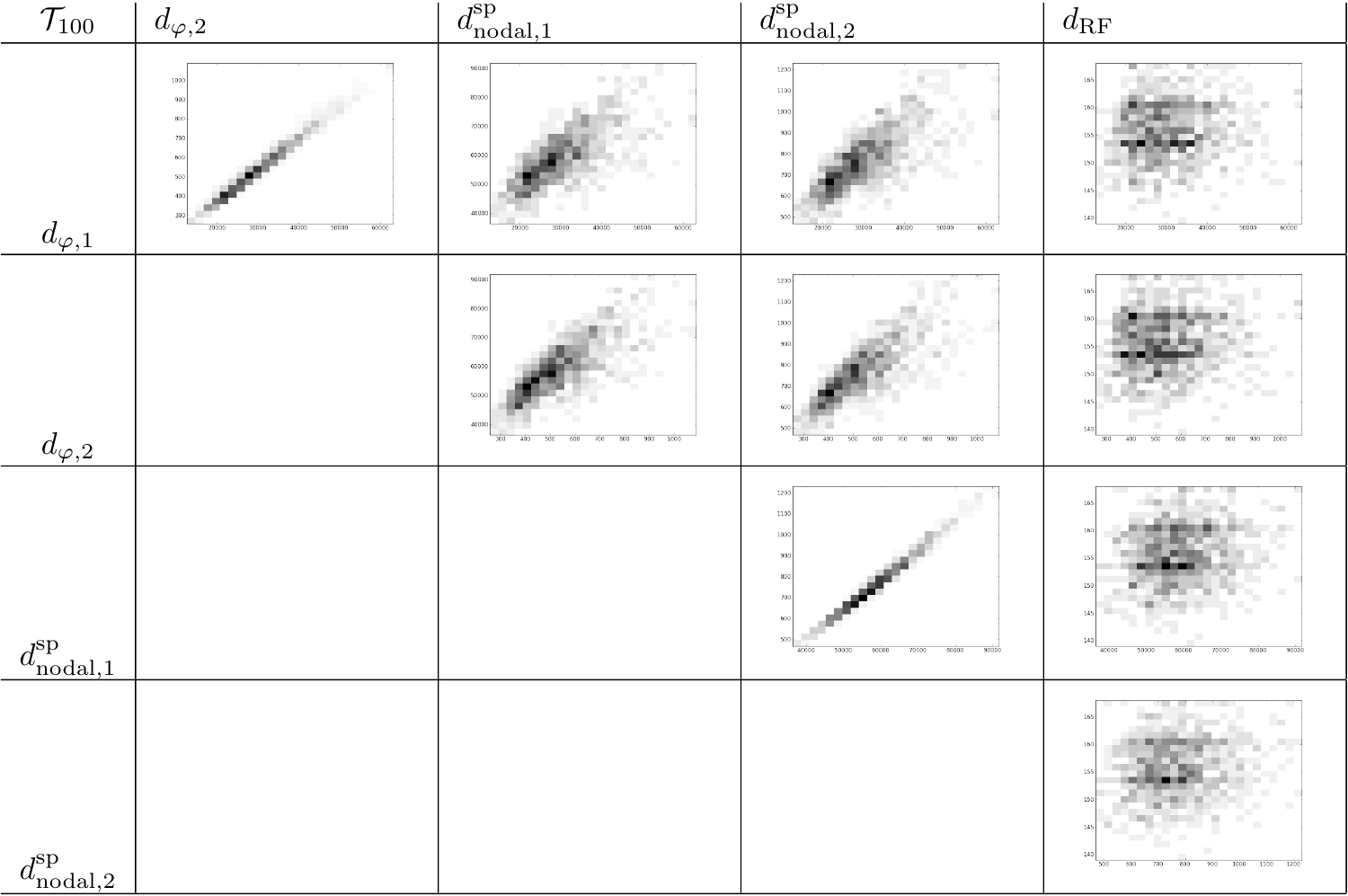}
\end{center}
\caption{\label{fig:2dhisto-nobin} 2D-histograms showing the relationship between different distances on $\TT_6$.}
\end{figure}
\end{center}

\begin{center}
\begin{table}[htb]
 \begin{center}
      \begin{tabular}{c|l|l|l|l|}
    $\BTT_{100}$ & $d_{\varphi,2}$ & $d_{\rm{nodal},1}$ & $d_{\rm{nodal},2}$ & $d_{\rm{RF}}$\\\hline
    $d_{\varphi,1}$ &0.986933 & 0.447140&0.448265 & -0.00080\\\hline
    $d_{\varphi,2}$ & &0.513306 &0.514363 &0.003281 \\\hline
    $d_{\rm{nodal},1}$ & & &0.998478 &0.012643\\\hline
    $d_{\rm{nodal},2}$ & & & & 0.012391\\\hline
  \end{tabular}
\end{center}
 \caption{\label{tab:spearman-bin} 2D-histograms showing the relationship between different
    distances on $\BTT_{100}$.}
\end{table}
\end{center}

\begin{center}
\begin{table}[htb]
 \begin{center}
    \begin{tabular}{c|l|l|l|l|}
    $\TT_{100}$ & $d_{\varphi,2}$ & $d_{\rm{nodal},1}^{\rm{sp}}$ & $d_{\rm{nodal},2}^{\rm{sp}}$ & $d_{\rm{RF}}$\\\hline
    $d_{\varphi,1}$ &0.987184 & 0.731755& 0.753918& 0.091556\\\hline
    $d_{\varphi,2}$ & &0.780030 &0.803423 &0.088390 \\\hline
    $d_{\rm{nodal},1}^{\rm{sp}}$ & & &0.990944 &0.132030 \\\hline
    $d_{\rm{nodal},2}^{\rm{sp}}$ & & & & 0.118336\\\hline
  \end{tabular}
\end{center}
 \caption{\label{tab:spearman-nobin} 2D-histograms showing the relationship between different
    distances on $\TT_{100}$.}
\end{table}
\end{center}

Finally, in Figure~\ref{fig:histograms} we have plotted the histograms of the
distributions of $d_{\varphi,1}$ and $d_{\varphi,2}$ on $\BTT_n$ and
$\TT_n$ for $n=10,20,\dots,100$. As it can be seen, they are positive skewed, like the splitted nodal metrics \cite[Fig. 5]{cardona.ea:08a}, but unlike other metrics like the Robinson-Foulds \cite{Steel88} or the transposition distance \cite[Fig. 2]{transdist}, which are negative skewed, or the  triples metric \cite{critchlow.ea:1996}, which is approximately normal.

 \begin{center}
\begin{figure}[p]
\begin{center}
\includegraphics[width=0.7\linewidth]{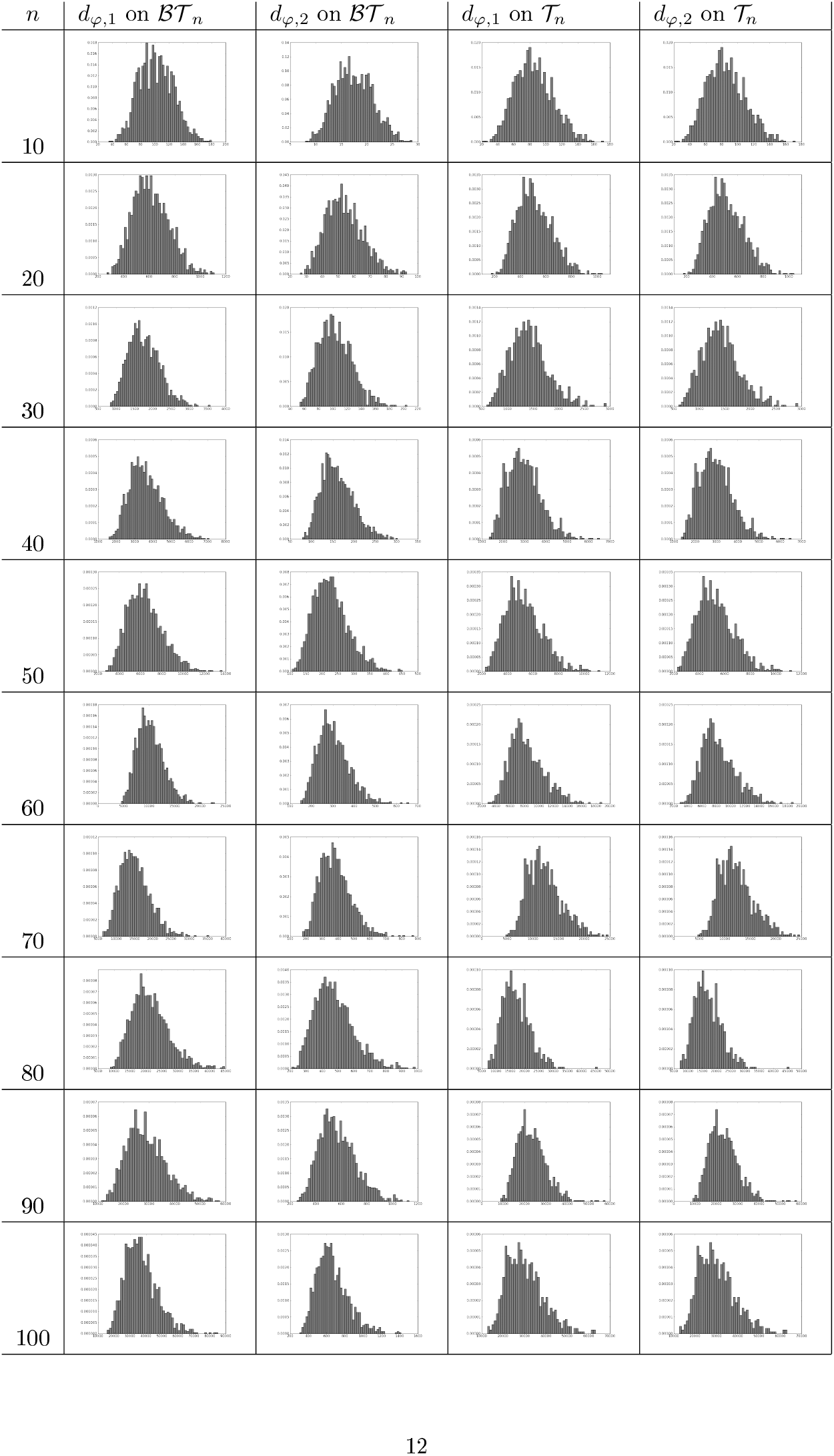}
\end{center}
\caption{\label{fig:histograms} Histograms of the
distributions of $d_{\varphi,1}$ and $d_{\varphi,2}$ on $\TT_n$ and
$\BTT_n$ for $n=10,20,\dots,100$.}
\end{figure}
\end{center}

\section*{Conclusions}

Following a fifty years old idea of Sokal and Rohlf \cite{sokal.roth:62}, we have encoded a weighted phylogenetic tree with nested taxa by means of its vector of cophenetic values of pairs of taxa, adding moreover to this vector the depths of single taxa. These positive real-valued vectors single out weighted phylogenetic trees with nested taxa, and therefore they can be used to define metrics to compare such trees. We have defined a family of metrics $d_{\varphi,p}$, for $p\in \{0\}\cup[1,\infty[$, by comparing these vectors through the $L^p$ norm.
 
We cannot advocate the use of any cophenetic metric $d_{\varphi,p}$ over
the other ones except, perhaps, warning against the use of the Hamming distance
$d_{\varphi,0}$ because it is too uninformative.
Since the most popular norms on $\RR^m$ are the Manhattan $L^1$ and the
Euclidean $L^2$, it seems natural to use $d_{\varphi,1}$ or $d_{\varphi,2}$. And since these two metrics are very highly correlated,  the comparison of trees using one or the other will not differ greatly. Each one of these metrics has its own advantages.
 
On the one hand, the computation of
$d_{\varphi,1}$ does not involve roots, and therefore it can be
computed exactly. Moreover, it takes integer values on unweighted trees and in this case its range of values is greater, thus being more discriminative.  
Actually, since 
$\|x\|_p\leq \|x\|_1$ for every $x\in \RR^m$ and $p\geq 1$,
we have that
$$
d_{\varphi,p}(T_1,T_2)\leq d_{\varphi,1}(T_1,T_2) \quad\mbox{ for every $T_1,T_2\in \WTT_n$}.
$$
On the other hand, the comparison of cophenetic vectors by
means of the Euclidean norm enables the use of many geometric
and clustering methods that are not available otherwise.  In particular, it is possible to compute the mean value of the square of $d_{\varphi,2}$ under different evolutionary models. We shall report on this elsewhere.

As a rule of thumb, and as we already advised in the context of splitted nodal metrics \cite{cardona.ea:08a}, we suggest using  $d_{\varphi,1}$ when the trees are unweighted, because these
trees can be seen as discrete objects and thus their comparison through a discrete
tool as the Manhattan norm seems appropriate.  When the trees have
arbitrary positive real weights, they should be understood as
belonging to a continuous space \cite{billera.ea:01}, and then the Euclidean norm is more
appropriate.

Future work will include a deeper study of the distribution of $d_{\varphi,1}$ and $d_{\varphi,2}$
on different spaces of unweighted phylogenetic trees.

\section*{Competing interests}
The authors declare that they have no competing interests.

\section*{Authors' contributions}
AM and FR developed the theoretical part of the paper. GC, LR and DS implemented the algorithms and performed the numerical experiments. GC and DS prepared the Supplementary Material web page. FR prepared the first version of the manuscript. All authors revised, discussed, and amended the manuscript. All authors read and approved the final manuscript. 

\section*{Acknowledgements}
The research reported in this paper has been partially supported by the Spanish government and the UE FEDER program, through project MTM2009-07165. We thank the comments and suggestions of the reviewers, which have led to a substantial improvement of this paper.

%

\section*{Appendix: Proofs of Propositions \ref{prop:minutt}--\ref{prop:diam}}
\label{app:min}

\subsection*{Proof of Proposition \ref{prop:minutt}}
By Lemma \ref{lem:pvs0}, it is enough to prove that the minimum non-zero value of $D_0$ is 1, and that all pairs $T,T'\in \UTT_n$ such that $D_0(T,T')=1$ also satisfy that $D_p(T,T')=1$ for every $p\geq 1$.

As we have seen in Example \ref{ex:contract}, if we contract a pendant arc in a tree $T$, we obtain a new tree $T'$ such that $D_p(T,T')=1$, for every $p\in \{0\}\cup[1,\infty[$, and this is of course the smallest possible non-negative value of $D_p$ on $\UTT_n$. It remains to prove that this is the only way we can obtain a pair of trees such that $D_0(T,T')=1$. 

So, let $T,T'\in \UTT_n$ be such that  $\varphi(T)=\varphi(T')+m\cdot e_{i,j}$ for some $m\geq 1$ and $1\leq i,j\leq n$ (where $e_{i,j}$ stands for the vector of length $n(n+1)/2$ with all entries 0 except an 1 in the entry corresponding to the pair $(i,j)$); that is, $T$ and $T'$ are such that $\fT(i,j)=\fTp(i,j)+m$,  for some $m\geq 1$, and  $\fT(x,y)=\fTp(x,y)$ for every $(x,y)\neq (i,j)$. Let us prove first of all that $m=1$. So, assume that $m\geq 2$ and let us reach a contradiction.

Since $\fT(i,j)>0$, there exists some taxon $k\neq i,j$ that is a descendant in $T$ of the parent of $[i,j]_T$. In other words, such that $[i,k]_T=[j,k]_T$ is the parent of $[i,j]_T$. But then
$$
\begin{array}{l}
\fTp(i,k)=\fT(i,k)=\fT(i,j)-1=\fTp(i,j)+(m-1)>\fTp(i,j)\\
\fTp(j,k)=\fT(j,k)=\fT(i,j)-1=\fTp(i,j)+(m-1)>\fTp(i,j)
\end{array}
$$
which cannot hold simultaneously: if $\fTp(i,k)>\fTp(i,j)$, then $\fTp(j,k)=\fTp(i,j)$. This shows that $m=1$, and thus $\varphi(T)=\varphi(T')+ e_{i,j}$.

Let us  prove now that it cannot happen that $i\neq j$. Indeed, assume that $i\neq j$. If $\varphi_{T'}(i,j)=\delta_{T'}(i)$, then 
$$
\varphi_T(i,j)=\fTp(i,j)+1=\delta_{T'}(i)+1=\delta_T(i)+1,
$$
 which is impossible. This implies that $\varphi_{T'}(i,j)<\delta_{T'}(i), \delta_{T'}(j)$. If, now, $\varphi_{T'}(i,j)<\delta_{T'}(i)-1$, then there will exist some leaf $k$ such that
$[i,k]_{T'}$ is the child of $[i,j]_{T'}$ in the path from $[i,j]_{T'}$ to $i$. Then
$\varphi_{T'}(i,k)=\varphi_{T'}(i,j)+1$ and $\varphi_{T'}(j,k)=\varphi_{T'}(i,j)$, which entail that
$$
\varphi_{T}(i,k)=\fTp(i,k)=\varphi_{T'}(i,j)+1=\varphi_{T}(i,j)>\varphi_{T'}(i,j)=\varphi_{T'}(j,k)=\varphi_{T}(j,k),
$$
 which is also impossible.
So, if $i\neq j$, the only possibility is that $\varphi_{T'}(i,j)=\delta_{T'}(i)-1=\delta_{T'}(j)-1$, but then
it would imply that $\varphi_{T}(i,j)=\fTp(i,j)+1=\delta_{T}(i)=\delta_{T}(j)$ and hence that $[i,j]_T=i=j$, which is again impossible.  

So, if $\varphi(T)=\varphi(T')+e_{i,j}$ then it must happen that $i=j$. In this case, moreover, $i$ must be a leaf in $T$ with unlabeled parent. Indeed, if $i$ is not a leaf, then there is some leaf $k$ such that $i=[i,k]_T$ and hence $\delta_T(i)=\varphi_T(i,k)$. Then,
$\delta_{T'}(i)=\delta_T(i)-1=\varphi_T(i,k)-1=\varphi_{T'}(i,k)-1$, which is impossible. So, $i$ is a leaf in $T$. And if its parent  is labeled, say with $l$, then
$\delta_T(i)=\delta_T(l)+1$ and $\delta_T(l)=\varphi_T(i,l)$. Thus, in $T'$,
$\delta_{T'}(i)=\delta_T(i)-1=\delta_T(l)=\delta_{T'}(l)$ and 
$\delta_{T'}(i)=\delta_T(l)=\varphi_T(i,l)=\varphi_{T'}(i,l)$, which is also impossible, since it would imply that $[i,l]_{T'}=i=l$.

So, finally,  it must happen that $i$ is a leaf in $T$ and its parent is not labeled. Let $T_0$ be the phylogenetic tree obtained from $T$ by contracting the pendant arc ending in $i$. Then $\varphi(T_0)=\varphi(T)-e_{i,i}=\varphi(T')$, and this implies, by Theorem \ref{th:char},  that $T_0= T'$.

This finishes the proof that the only pairs $T,T'\in \WTT_n$ such that $D_0(T,T')=1$ are those where one of them is obtained from the other by the contraction of a pendant arc. Since these pairs of trees also satisfy that $D_p(T,T')=1$ for every $p\geq 1$, this completes the proof of the proposition.
\qed

\subsection*{Proof of Proposition \ref{prop:mintt}}
To ease the task of the reader, we split this proof into several lemmas. To begin with, notice that there are pairs of trees $T,T'\in \TT_n$ such that $D_p(T,T') =3$ for every $p\in \{0\}\cup[1,\infty[$: for instance, by Example \ref{ex:contract}, when $T'$ is obtained from $T$ by contracting an arc ending in the root of a cherry. So, the minimum non-zero value of $D_p(T,T')$ on $\TT_n$ is at most $3$.

\begin{lemma}\label{lem:mtt0}
If $T,T'\in \TT_n$ are such that $D_0(T,T')>0$, then there exists a pair of different taxa $i\neq j$ such that
$\varphi_T(i,j)\neq \varphi_{T'}(i,j)$.
\end{lemma}

\begin{proof}
If $\varphi_T(i,j)= \varphi_{T'}(i,j)$ for every $i\neq j$, then, by Corollary \ref{cor:char}, $T= T'$ and therefore $D_0(T,T')=0$.
\end{proof}

So, every pair of  phylogenetic trees in $\TT_n$ at non-zero $D_0$ distance must have a pair of different leaves with different cophenetic values.

\begin{lemma}\label{lem:mtt1}
Let $T,T'\in \TT_n$ be such that $\varphi_T(i,j)=\varphi_{T'}(i,j)+m$, for some $1\leq i< j\leq n$ and some $m\geq 1$. Let $k\neq i,j$ be a leaf such that
there exists a path from $[i,j]_{T'}$ to $[i,k]_{T'}$ of length $l$,  for some $l\geq 1$.  Then:
\begin{enumerate}[(a)]
\item If $\varphi_T(i,k)=\varphi_{T'}(i,k)$, then $\varphi_T(j,k)\geq \varphi_{T'}(j,k)+\min\{m,l\}$
\item If $\varphi_T(j,k)=\varphi_{T'}(j,k)$, then $\varphi_T(i,k)= \varphi_{T'}(i,k)-l$
\end{enumerate}
\end{lemma}

\begin{proof}
From the assumptions we have that  $\varphi_{T'}(i,k)=\varphi_{T'}(i,j)+l=\fTp(j,k)+l$. Now:
\begin{enumerate}[(a)]
\item Assume that $\varphi_T(i,k)=\varphi_{T'}(i,k)$. Then,
$$
\fT(i,k)=\fTp(i,k)=\fTp(i,j)+l=\fT(i,j)-(m-l),
$$
and then
\begin{itemize}
\item If $m>l$, then $\fT(i,k)<\fT(i,j)$, that is, $[i,j]_T\prec[i,k]_T$, and thus
$$
\fT(j,k)=\fT(i,k)=\fTp(i,k)=\fTp(j,k)+l
$$
\item If $m=l$, then $\fT(i,k)=\fT(i,j)$, that is, $[i,k]_T=[i,j]_T$, and thus
$$
\fT(j,k)\geq \fT(i,j)=\fTp(i,j)+m=\fTp(j,k)+m
$$
\item If $m<l$, then $\fT(i,k)>\fT(i,j)$, that is, $[i,k]_T\prec [i,j]_T$, and thus
$$
\fT(j,k)= \fT(i,j)=\fTp(i,j)+m=\fTp(j,k)+m
$$
\end{itemize}

\item Assume that $\varphi_T(j,k)=\varphi_{T'}(j,k)$. Then
$$
\fT(j,k)=\fTp(j,k)=\fTp(i,j)=\fT(i,j)-m,
$$
so that $[i,j]_T\prec [j,k]_T$, and thus
$$
\fT(i,k)=\fT(j,k)=\fTp(j,k)=\fTp(i,j)=\fTp(i,k)-l
$$
\end{enumerate}
\end{proof}

As a direct consequence of this lemma we obtain the following result.

\begin{corollary}\label{cor:mtt1}
Let $T,T'\in \TT_n$ be such that $\varphi_T(i,j)=\varphi_{T'}(i,j)+m$, for some $1\leq i< j\leq n$ and some $m\geq 1$. Let $N$ be the number of leaves $k$ such that $k\neq i,j$ and either
$[i,k]_{T'}\prec [i,j]_{T'}$ or $[j,k]_{T'}\prec [i,j]_{T'}$. Then,
$$
D_0(T,T')\geq N+1.
$$
\ \qed
\end{corollary}

\begin{lemma}\label{lem:mtt2}
Let $T,T'\in \TT_n$ be such that $D_0(T,T')\leq 3$. If $\varphi_T(i,j)=\varphi_{T'}(i,j)+m$, for some $1\leq i< j\leq n$ and some $m\geq 1$, then $m=1$.
\end{lemma}

\begin{proof}
If $\delta_{T'}(i)=\delta_T(i)$, then $\delta_{T'}(i)=\delta_T(i)>\varphi_T(i,j)=\fTp(i,j)+m$ which implies that there are at least $m$ leaves $k$ such that $[i,k]_{T'}\prec [i,j]_{T'}$. Then, by the last corollary,
$D_0(T,T')\geq m+1$. Now, if $\delta_{T'}(j)=\delta_T(j)$, then for the same reason 
there are at least $m$ leaves $k$ such that $[j,k]_{T'}\prec [i,j]_{T'}$ and they increase $D_0(T,T')$ to at least $2m+1$, while if $\delta_{T'}(j)\neq \delta_T(j)$, then $D_0(T,T')\geq m+2$. 
We conclude then that if $\delta_{T'}(i)=\delta_T(i)$, then $m=1$.
By symmetry, if  $\delta_{T'}(j)=\delta_T(j)$, then $m=1$, either. 

Finally, if  $\delta_{T'}(i)\neq \delta_T(i)$ and  $\delta_{T'}(j)\neq \delta_T(j)$, and since 
$\varphi_T(i,j)\neq \varphi_{T'}(i,j)$, we have that 
$\fT(x,y)=\fTp(x,y)$ for every $(x,y)\neq (i,i),(j,j),(i,j)$. Let now  $k\neq i,j$ be a taxon such that $[i,k]_T=[j,k]_T$ is the parent of $[i,j]_T$ in $T$. Then
$$
\fTp(i,k)=\fT(i,k)=\fT(i,j)-1=\fTp(i,j)+(m-1)
$$
and therefore, if $m\geq 2$, $\fTp(i,k)>\fTp(i,j)$ and then, by
Lemma \ref{lem:mtt1}, either 
$\fT(i,k)\neq \fTp(i,k)$ or $\fT(j,k)\neq \fTp(j,k)$, which, as we have seen, is impossible. Thus, $m=1$ in all cases.
\end{proof}

\begin{lemma}\label{lem:mtt3}\label{lem:mtt4}
Let $T,T'\in \TT_n$ be such that $D_0(T,T')\leq 3$. If $\varphi_T(i,j)=\varphi_{T'}(i,j)+1$, for some $1\leq i< j\leq n$, then $(\delta_{T'}(i)- \varphi_{T'}(i,j))+(\delta_{T'}(j)- \varphi_{T'}(i,j))\leq 3$.
\end{lemma}

\begin{proof}
Let us assume that $(\delta_{T'}(i)- \varphi_{T'}(i,j))+(\delta_{T'}(j)- \varphi_{T'}(i,j))\geq 4$ and let us reach a contradiction.

Assume first that $\delta_{T'}(i)\geq \varphi_{T'}(i,j)+3$. Then, there are at least two leaves $k_1,k_2$ such that 
$[i,k_1]_{T'},[i,k_2]_{T'}\prec [i,j]_{T'}$. Since each such leaf contributes at least 1 to $D_0(T,T')\leq 3$, we conclude that there must be exactly two such leaves and, moreover, $\varphi_T(x,y)=\fTp(x,y)$ for every $(x,y)\neq (i,j),(i,k_1),(j,k_1),(i,k_2),(j,k_2)$. But then, on the one hand, $\delta_T(j)=\delta_{T'}(j)$ and, on the other hand,
$\delta_{T'}(j)=\fTp(i,j)+1$ (otherwise, there would be some other leaf $k$ such that $[j,k]_{T'}\prec [i,j]_{T'}$, which, by Lemma \ref{lem:mtt1} would satisfy that $\varphi_T(i,k)\neq \fTp(i,k)$ or $\varphi_T(j,k)\neq \fTp(j,k)$). Combining these two equalities we obtain 
$\delta_T(j)=\fT(i,j)$,
 which is impossible in a tree without nested taxa. 
  This proves that
$\delta_{T'}(i) \leq \varphi_{T'}(i,j)+2$ and, by symmetry, that 
$\delta_{T'}(j) \leq \varphi_{T'}(i,j)+2$, as we claimed. 

Thus, it remains to prove that the case $\delta_{T'}(i)=\delta_{T'}(j)= \varphi_{T'}(i,j)+2$ is impossible. 
So, assume this case holds, and let's reach a contradiction. 
By Corollary \ref{cor:mtt1},  if $D_0(T,T')\leq 3$ and $\delta_{T'}(i)=\delta_{T'}(j)= \varphi_{T'}(i,j)+2$, then there can exist only one extra leaf $k$ pending from the parent of $i$ and one extra leaf $l$ pending from the parent of $j$: see Fig. \ref{fig:mtt41}, where the grey triangle stands for the (possibly empty) subtree consisting of all other descendants of $[i,j]_{T'}$.  Moreover, since $\varphi_T(i,j)=\varphi_{T'}(i,j)+1$ and since both $k$ and $l$ contribute at least 1 to $D_0(T,T')\leq 3$, we conclude that $\varphi_T(x,y)=\fTp(x,y)$ for every $(x,y)\neq (i,j),(i,k),(j,k),(i,l),(j,l)$. In particular
$$
\begin{array}{l}
\fT(k,l)=\fTp(k,l)=\fTp(i,j)=\fT(i,j)-1\\
\delta_T(i)=\delta_{T'}(i)=\fTp(i,j)+2=\fT(i,j)+1\\
\delta_T(j)=\delta_T(k)=\delta_T(l)=\fT(i,j)+1\mbox{ for the same reason}
\end{array}
$$

\begin{center}
\begin{figure}[htb]
\begin{center}
\begin{tikzpicture}[thick,>=stealth,scale=0.3]
\draw(0,0) node [tre] (i) {};  \etq i
\draw(2,0) node [tre] (k) {};  \etq k
\draw(4,0) node [tre] (l) {};  \etq l
\draw(6,0) node [tre] (j) {};  \etq j  
\draw(1,2) node [tre] (a) {};   
\draw(5,2) node [tre] (b) {};   
\draw(3,5) node [tre] (r) {};   
\draw(3.5,6.5) node {\footnotesize $[i,j]_{T'}$};
\draw[fill=black!15] (r)--(8,1)--(14,1)--(r);
\draw (r)--(b);
\draw (r)--(a);
\draw (a)--(i);
\draw (a)--(k);
\draw (b)--(l);
\draw (b)--(j);
\end{tikzpicture}
\end{center}
\caption{\label{fig:mtt41} 
The subtree of $T'$ rooted at $[i,j]_{T'}$ in the proof of Lemma \ref{lem:mtt4}.}
\end{figure}
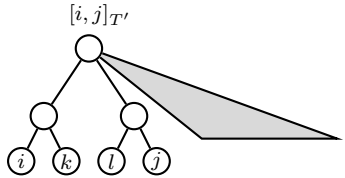
\end{center}

Now we shall prove that, in this situation, each one of $k,l$ contributes actually at least 2 to $D_0(T,T')$, and therefore $D_0(T,T')\geq 5$, which contradicts the assumption that $D_0(T,T')\leq 3$ .

\begin{enumerate}[(1)]
\item Assume that $\fT(i,k)=\fTp(i,k)$. Then, by Lemmas \ref{lem:mtt1}.(a) and \ref{lem:mtt2}, $\fT(j,k)=\fTp(j,k)+1$, and hence
$$
\begin{array}{l}
\fT(i,k)=\fTp(i,k)=\fTp(j,k)+1=\fT(j,k)\\
\fT(i,k)=\fTp(i,k)=\fTp(i,j)+1=\fT(i,j)\\ \delta_T(i)=\delta_T(j)=\delta_T(k)=\delta_T(l)=\fT(i,j)+1\\
 \fT(k,l)=\fT(i,j)-1
\end{array}
$$
Thus, the subtree of $T$ rooted at  $[k,l]_{T}$ contains a subtree of the form described in Fig. \ref{fig:mtt42}, for at least one leaf $h$. But then
$$
\fTp(l,h)=\fT(l,h)=\fT(i,j)=\fTp(i,j)+1=\varphi_{T'}(l,j)
$$
which is impossible, since it would imply that $h$ is another descendant of $[l,j]_{T'}$. Therefore, $\fT(i,k)\neq \fTp(i,k)$ and, by symmetry,  $\fT(j,l)\neq \fTp(j,l)$.

\begin{center}
\begin{figure}[htb]
\begin{center}
\begin{tikzpicture}[thick,>=stealth,scale=0.3]
\draw(0,0) node [tre] (i) {};  \etq i
\draw(2,0) node [tre] (k) {};  \etq k
\draw(4,0) node [tre] (j) {};  \etq j
\draw(6,0) node [tre] (h) {};  \etq h  
\draw(8,0) node [tre] (l) {};  \etq l  
\draw(2,3) node [tre] (a) {};   
\draw(7,3) node [tre] (b) {};   
\draw(5,6) node [tre] (r) {};   
\draw(5,7.5) node {\footnotesize $[k,l]_{T}$};
\draw (r)--(b);
\draw (r)--(a);
\draw (a)--(i);
\draw (a)--(k);
\draw (a)--(j);
\draw[snake=coil,segment aspect=0] (b)--(h);
\draw (b)--(l);
\end{tikzpicture}
\end{center}
\caption{\label{fig:mtt42} 
A subtree of the subtree of $T$ rooted at $[k,l]_{T}$ in case (1) in the proof of Lemma \ref{lem:mtt4}.}
\end{figure}
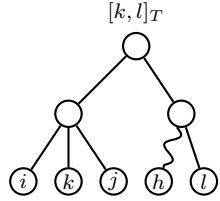
\end{center}

\item  Assume now that $\fT(i,l)=\fTp(i,l)$. Then, by Lemma \ref{lem:mtt1}.(b), $\fT(j,l)=\fTp(j,l)-1$, and then
$$
\begin{array}{l}
\fT(i,l)=\fTp(i,l)=\fTp(i,j)=\fT(i,j)-1\\
\fT(j,l)=\fTp(j,l)-1=\fTp(i,j)=\fT(i,j)-1\\
\fT(k,l)=\fT(i,j)-1\\
\delta_T(i)=\delta_T(j)=\delta_T(k)=\delta_T(l)=\fT(i,j)+1
\end{array}
$$
Therefore, the subtree of $T$ rooted at  $[k,l]_{T}$ contains a subtree of the form  described  
in Fig. \ref{fig:mtt43}, for at least one leaf $h$. Moreover, $h\neq k$ because $\fT(h,l)>\fT(j,l)=\fT(k,l)$.
 But then, again,
$$
\fTp(l,h)=\fT(l,h)=\fT(i,j)=\fTp(i,j)+1=\varphi_{T'}(l,j)
$$
which is again impossible by the same reason as in (1). Therefore, $\fT(i,l)\neq \fTp(i,l)$ and, by symmetry, 
$\fT(j,k)\neq \fTp(j,k)$.

\begin{center}
\begin{figure}[htb]
\begin{center}
\begin{tikzpicture}[thick,>=stealth,scale=0.3]
\draw(0,0) node [tre] (i) {};  \etq i
\draw(2,0) node [tre] (j) {};  \etq j
\draw(4,0) node [tre] (h) {};  \etq h  
\draw(6,0) node [tre] (l) {};  \etq l  
\draw(-6,0) node [tre] (k) {};  \etq k  
\draw(1,3) node [tre] (a) {};   
\draw(5,3) node [tre] (b) {};   
\draw(3,6) node [tre] (r) {};   
\draw[snake=coil,segment aspect=0] (r)--(k);
\draw(3,7.5) node {\footnotesize $[k,l]_{T}$};
\draw (r)--(b);
\draw (r)--(a);
\draw (a)--(i);
\draw (a)--(j);
\draw[snake=coil,segment aspect=0] (b)--(h);
\draw (b)--(l);
\end{tikzpicture}
\end{center}
\caption{\label{fig:mtt43} 
A subtree of the subtree of $T$ rooted at $[k,l]_{T}$ in case (2) in the proof of Lemma \ref{lem:mtt4}.}
\end{figure}
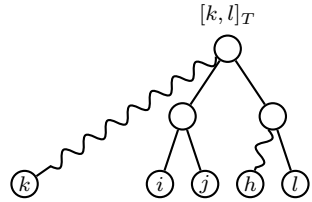
\end{center}
\end{enumerate}
So, 
$$
\fT(i,k)\neq \fTp(i,k), \fT(i,l)\neq \fTp(i,l), \fT(j,k)\neq \fTp(j,k), \fT(j,l)\neq \fTp(j,l)
$$
 and thus $D_0(T,T')\geq 5$.
\end{proof}

Summarizing the last lemmas, we have proved so far that if $D_0(T,T')\leq 3$ and $\varphi_T(i,j)\neq \varphi_{T'}(i,j)$, then, up to interchanging $T$ and $T'$, $\varphi_T(i,j)= \varphi_{T'}(i,j)+1$ and 
either  $i$ and $j$ are sibling in $T'$ or one of these  leaves is a sibling of the parent of the other one in $T'$. Next two lemmas cover these two remaining cases.

\begin{lemma}\label{lem:mtt6}
Let $T,T'\in \TT_n$ be such that $D_0(T,T')\leq 3$, and assume that $\varphi_T(i,j)=\varphi_{T'}(i,j)+1$, for some $1\leq i< j\leq n$. If $i$ and $j$ are sibling in $T'$, then they are also sibling in $T$,  they have no other sibling in $T$, and $T'$ is obtained from $T$ by contracting the arc ending in $[i,j]_T$. And then, $D_0(T,T')= 3$.
\end{lemma}

\begin{proof}
If $\delta_{T'}(i)=\delta_{T'}(j)=\varphi_{T'}(i,j)+1$, then it must happen that $\delta_T(i)=\delta_{T'}(i)+1$ and 
$\delta_T(j)=\delta_{T'}(j)+1$. Indeed, if $\delta_T(i)\leq \delta_{T'}(i)$, then
$\delta_T(i)\leq \varphi_{T'}(i,j)+1=\varphi_T(i,j)$, which is impossible. Therefore, $\delta_T(i)> \delta_{T'}(i)$ and by symmetry $\delta_T(j)> \delta_{T'}(j)$. Since $\varphi_T(i,j)=\varphi_{T'}(i,j)+1$, $D_0(T,T')\leq 3$ implies that $\fT(x,y)=\fTp(x,y)$, for every $(x,y)\neq (i,j), (i,i), (j,j)$. Now,
if, say $\delta_T(i)\geq \delta_{T'}(i)+2$, then
 $$
 \delta_T(i)\geq \delta_{T'}(i)+2=\fTp(i,j)+3=\fT(i,j)+2
 $$
 and there would exist some leaf $k$ such that $[i,k]_T$ is a child of $[i,j]_T$. But then
 $$
 \fTp(i,k)=\fT(i,k)=\fT(i,j)+1=\fTp(i,j)+2=\delta_{T'}(i)+1,
  $$
which is impossible. This proves that $\delta_T(i)=\delta_{T'}(i)+1$ and, by symmetry, $\delta_T(j)=\delta_{T'}(j)+1$.

So, in summary, $\varphi_T(i,j)=\varphi_{T'}(i,j)+1$, $\delta_T(i)=\delta_{T'}(i)+1$, $\delta_T(j)=\delta_{T'}(j)+1$ and $\fT(x,y)=\fTp(x,y)$, for every $(x,y)\neq (i,j), (i,i), (j,j)$, and in particular $d_{\varphi,p}(T,T')=3$.

Now, $\delta_T(i)=\delta_{T'}(i)+1=\fTp(i,j)+2=\fT(i,j)+1$,
and by symmetry, $\delta_T(j)=\fT(i,j)+1$, either. Therefore, $i$ and $j$ are sibling in $T$. Let us see that they have no other sibling in this tree. Indeed, if $k$ is a sibling of $i$ and $j$ in $T$, then
$$
\fTp(i,k)=\fT(i,k)=\fT(i,j)=\fTp(i,j)+1=\delta_{T'}(i)
$$
which is impossible.

Let $x$ be the parent of $[i,j]_T$, and assume that the subtree $T_0$ of $T$ rooted at $x$ is as described in Fig. \ref{figcasg:1}.(a), for some (possibly empty) subtree $\widehat{T}$. Moreover, let $T_0'$ be the subtree of $T'$ rooted at $[i,j]_{T'}$, which is as described in Fig. \ref{figcasg:1}.(b) for some subtree $\widehat{T'}$.
We shall prove that $\widehat{T}=\widehat{T'}$.
\begin{center}
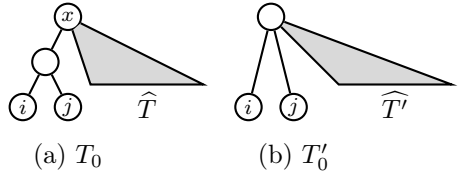
\begin{figure}[htb]
\begin{center}
\begin{tikzpicture}[thick,>=stealth,scale=0.3]
  \begin{scope}
\draw(0,0) node [tre] (i) {};  \etq i
\draw(2,0) node [tre] (j) {};  \etq j 
\draw(1,2) node [tre] (a) {};   
\draw(2,4) node [tre] (x) {};  \etq x 
\draw (x)--(a);
\draw (a)--(j);
\draw (a)--(i);
\draw[fill=black!15] (x)--(3,1)--(8,1)--(x);
\draw(5.5,0.1) node   {$\widehat{T}$};  
\draw(2,-2.2) node   {(a) $T_0$};
\end{scope}
\begin{scope}[xshift=10cm]
\draw(0,0) node [tre] (i) {};  \etq i
\draw(2,0) node [tre] (j) {};  \etq j 
\draw(1,4) node [tre] (a) {};   
\draw (a)--(j);
\draw (a)--(i);
\draw[fill=black!15] (a)--(4,1)--(9,1)--(a);
\draw(6.5,0.1) node   {$\widehat{T'}$};  
\draw(2,-2.2) node   {(b) $T'_0$};
\end{scope}
\end{tikzpicture}
\end{center}
\caption{\label{figcasg:1}
(a) The subtree $T_0$ of $T$ rooted at the parent of $[i,j]_{T}$ in the proof of Lemma \ref{lem:mtt6}.
(b) The subtree $T_0'$ of $T'$ rooted at $[i,j]_{T'}$ in the proof of the same Lemma.}
\end{figure}
\end{center}

For every $k\in L(\widehat{T})$,
$$
\fTp(i,k)=\fT(i,k)=\fT(i,j)-1=\fTp(i,j), 
$$
which entails that $k\in L(\widehat{T'})$. Conversely, if 
$k\in L(\widehat{T'})$, then
$$
\fT(i,k)=\fTp(i,k)=\fTp(i,j)=\fT(i,j)-1, 
$$
which entails that  $k\in L(\widehat{T})$. Thus,  $L(\widehat{T})=L(\widehat{T'})$. And finally,  for every (not necessarily different) $k,l\in L(\widehat{T})$, 
$$
\varphi_{\widehat{T}}(k,l)=\fT(k,l)-\delta_T(x)=
\fT(k,l)-\fT(i,j)+1=\fTp(k,l)-\fTp(i,j)=\varphi_{\widehat{T'}}(k,l),
$$
which implies by Theorem \ref{th:char} that $\widehat{T}= \widehat{T'}$ (notice that $\widehat{T}$ and  $\widehat{T'}$ can have elementary roots).

Finally, let us prove now that $T$ and $T'$ are exactly the same except for $T_0$ and $T_0'$. 
More specifically, let $T_1$ and $T_1'$ be obtained by replacing in $T$ and $T'$ the subtrees $T_0$ and $T_0'$ by a single leaf $x$. Since for every $p,q\notin L(T_0)=L(T_0')$,
$$
\begin{array}{l}
\varphi_{T_1'}(p,q)=\fTp(p,q)=\fT(p,q)=\varphi_{T_1}(p,q),\\
\varphi_{T_1'}(x,p)=\fTp(i,p)=\fT(i,p)=\varphi_{T_1}(p,x),
\end{array}
$$
we deduce, again  by Theorem \ref{th:char}, that $T_1=T_1'$. 

This completes the proof that  $T'$ is obtained from $T$ by replacing in it the subtree $T_0$ rooted at the parent $x$ of $[i,j]_T$ by the subtree $T_0'$ obtained from $T_0$ by contracting the arc $(x,[i,j]_T)$.
\end{proof}

\begin{lemma}\label{lem:mtt5}
Let $T,T'\in \TT_n$ be such that $D_0(T,T')\leq 3$. Assume that $\varphi_T(i,j)=\varphi_{T'}(i,j)+1$, for some $1\leq i< j\leq n$, and that $j$ is a sibling of the parent of $i$ in $T'$. Then, the subtree of $T'$ rooted at $[i,j]_{T'}$ is the tree $T_0'$ depicted in Fig. \ref{fig:mtt5}.(a), for some taxon $k\neq i,j$ and some (possibly empty) subtree $\widehat{T'}$, and $T$ is obtained from $T'$ by replacing $T_0'$ by the tree $T_0$ depicted in Fig. \ref{fig:mtt5}.(b). And then, $D_0(T,T')= 3$.
\end{lemma}

\begin{center}
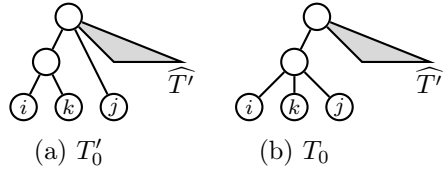
\begin{figure}[htb]
\begin{center}
\begin{tikzpicture}[thick,>=stealth,scale=0.3]
  \begin{scope}
\draw(0,0) node [tre] (i) {};  \etq i
\draw(2,0) node [tre] (k) {};  \etq k
\draw(4,0) node [tre] (j) {};  \etq j 
\draw(1,2) node [tre] (a) {};   
\draw(2,4) node [tre] (x) {};  
\draw (x)--(a);
\draw (x)--(j);
\draw (a)--(i);
\draw (a)--(k);
\draw[fill=black!15] (x)--(4,2)--(7,2)--(x);
\draw(7,1.2) node   {$\widehat{T'}$};  
\draw(2,-2) node   {(a) $T_0'$};
\end{scope}
\begin{scope}[xshift=10cm]
\draw(0,0) node [tre] (i) {};  \etq i
\draw(2,0) node [tre] (k) {};  \etq k
\draw(4,0) node [tre] (j) {};  \etq j 
\draw(2,2) node [tre] (a) {};   
\draw(3,4) node [tre] (x) {};  
\draw (x)--(a);
\draw (a)--(j);
\draw (a)--(i);
\draw (a)--(k);
\draw[fill=black!15] (x)--(5,2)--(8,2)--(x);
\draw(8,1.2) node   {$\widehat{T'}$};  
\draw(2,-2) node   {(b) $T_0$};
\end{scope}
\end{tikzpicture}
\end{center}
\caption{\label{fig:mtt5} 
(a) The subtree $T_0'$ of $T'$ rooted at $[i,j]_{T'}$ in the statement of Lemma \ref{lem:mtt5}.
(b) The subtree $T_0$ which replaces $T_0'$ in $T$ in the same statement.}
\end{figure}
\end{center}

\begin{proof}
We assume that $\delta_{T'}(i)=\varphi_{T'}(i,j)+2$ and $\delta_{T'}(j)=\varphi_{T'}(i,j)+1$. This implies that there exists at least one leaf $k$ such that $[i,k]_{T'}\prec [i,j]_{T'}$.
Since $\varphi_T(i,j)=\varphi_{T'}(i,j)+1$, $|\fT(i,k)-\fTp(i,k)|+|\fT(j,k)-\fTp(j,k)|\geq 1$ and $\delta_T(j)>\delta_{T'}(j)$ (because, otherwise,
$\delta_T(j)\leq \delta_{T'}(j)=\varphi_{T'}(i,j)+1=\varphi_T(i,j)$, which is impossible), $D_0(T,T')\leq 3$ 
entails that  $\fT(i,k)=\fTp(i,k)$ or $\fT(j,k)=\fT(j,k)$, and that  $\fT(x,y)=\fTp(x,y)$ for every $(x,y)\neq (i,j), (i,k), (j,k), (j,j)$ (and, in particular, $k$ is the only leaf different from $i$ such that $[i,k]_{T'}\prec [i,j]_{T'}$). Moreover, we have that $D_0(T,T')= 3$. 

Let us see now that $\delta_T(j)=\delta_{T'}(j)+1$. Indeed, if 
$\delta_T(j)\geq \delta_{T'}(j)+2$, then
 $$
 \delta_T(j)\geq \delta_{T'}(j)+2=\fTp(i,j)+3=\fT(i,j)+2
 $$
 and there would exist some leaf $l$ such that $[j,l]_T$ is a child of $[i,j]_T$. But then
 $$
 \fTp(j,l)=\fT(j,l)=\fT(i,j)+1=\fTp(i,j)+2=\delta_{T'}(j)+1
  $$
 and we reach a contradiction.

So, in summary, the subtree $T_0'$ of $T'$ rooted a $[i,j]_{T'}$ is as described in Fig. \ref{fig:mtt5}.(a), and
$\fT(i,j)=\fTp(i,j)+1$, $\delta_T(j)=\delta_{T'}(j)+1$,   $\fT(x,y)=\fTp(x,y)$ for every $(x,y)\neq (i,j), (i,k), (j,k), (j,j)$, and either $\fT(i,k)=\fTp(i,k)$ or $\fT(j,k)=\fT(j,k)$. Now, we discuss these two possibilities.

\begin{enumerate}[(a)]
\item If $\fT(j,k)=\fTp(j,k)$, then $\fT(i,k)=\fTp(i,k)-1$ by Lemma \ref{lem:mtt1}.(b).
In this case
$$
\begin{array}{l}
\fT(i,k)=\fTp(i,k)-1=\fTp(i,j)=\fT(i,j)-1\\
\fT(j,k)=\fTp(j,k)=\fTp(i,j)=\fT(i,j)-1\\
\delta_T(i)=\delta_{T'}(i)=\fTp(i,j)+2=\fT(i,j)+1\\
\delta_T(j)=\delta_{T'}(j)+1=\fTp(i,j)+2=\fT(i,j)+1\\
\delta_T(k)=\delta_{T'}(k)=\fTp(i,j)+2=\fT(i,j)+1
\end{array}
$$
This means that the subtree of $T$ rooted at $[i,k]_T=[j,k]_T$ contains a subtree of the  form described in Fig. \ref{fig:mtt5.1}, for at least some new leaf $h$.
But then
$$
\fTp(k,h)=\fT(k,h)=\fT(i,j)=\fTp(i,j)+1=\fTp(i,k)
$$
which is impossible in $T'$, because $i$ and $k$ are the only descendants of $[i,k]_{T'}$ in $T'$. So, this case is impossible.
\begin{center}
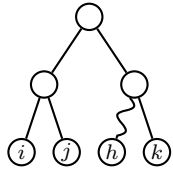
\begin{figure}[htb]
\begin{center}
\begin{tikzpicture}[thick,>=stealth,scale=0.3]
\draw(0,0) node [tre] (i) {};  \etq i
\draw(2,0) node [tre] (j) {};  \etq j
\draw(4,0) node [tre] (h) {};  \etq h  
\draw(6,0) node [tre] (k) {};  \etq k  
\draw(1,3) node [tre] (a) {};   
\draw(5,3) node [tre] (b) {};   
\draw(3,6) node [tre] (r) {};   
\draw (r)--(b);
\draw (r)--(a);
\draw (a)--(i);
\draw (a)--(j);
\draw[snake=coil,segment aspect=0] (b)--(h);
\draw (b)--(k);
\end{tikzpicture}
\end{center}
\caption{\label{fig:mtt5.1} 
A subtree contained in the subtree of $T$ rooted at $[i,j]_T$ in case (a) in the proof of Lemma \ref{lem:mtt5}.}
\end{figure}
\end{center}

\item If $\fT(i,k)=\fTp(i,k)$, then $\fT(j,k)=\fTp(j,k)+ 1$ Lemmas \ref{lem:mtt1}.(a) and \ref{lem:mtt2}. In this case
$$
\begin{array}{l}
\fT(i,k)=\fTp(i,k)=\fTp(i,j)+1=\fT(i,j)\\
\fT(j,k)=\fTp(j,k)+1=\fTp(i,j)+1=\fT(i,j)\\
\delta_T(i)=\delta_T(j)=\delta_T(k)=\fT(i,j)+1\mbox{ as in (a)}
\end{array}
$$
This implies that $i,j,k$ are sibling in $T$. If $l$ is any other sibling of them in $T$, then
$$
\fTp(i,l)=\fT(i,l)=\fT(i,k)=\fTp(i,k)
$$
which entails that $l$ is another descendant  of $[i,k]_{T'}$ in $T'$, which is impossible. Therefore, 
the subtree $T_0$ of $T$ rooted at the parent of $[i,j]_T$ has the form depicted in Fig. \ref{fig:mtt5.b}, for some subtree $\widehat{T}$. 

Finally, the same argument as in the last part of the proof of the last lemma shows that  $\widehat{T}=\widehat{T'}$, and that if  $T_1$ and $T_1'$ are obtained by replacing in $T$ and $T'$ the subtrees $T_0$ and $T_0'$ by a single leaf $x$, then $T_1= T_1'$. We leave the details to the reader.
\end{enumerate}
\begin{center}
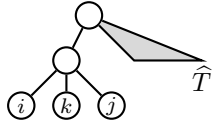
\begin{figure}[htb]
\begin{center}
\begin{tikzpicture}[thick,>=stealth,scale=0.3]
\draw(0,0) node [tre] (i) {};  \etq i
\draw(2,0) node [tre] (k) {};  \etq k
\draw(4,0) node [tre] (j) {};  \etq j 
\draw(2,2) node [tre] (a) {};   
\draw(3,4) node [tre] (x) {};  
\draw (x)--(a);
\draw (a)--(j);
\draw (a)--(i);
\draw (a)--(k);
\draw[fill=black!15] (x)--(5,2)--(8,2)--(x);
\draw(8,1.2) node   {$\widehat{T}$};  
\end{tikzpicture}
\end{center}
\caption{\label{fig:mtt5.b} 
The subtree $T_0$ rooted at the parent of $[i,j]_T$ in case (b) in the proof of Lemma \ref{lem:mtt5}.}
\end{figure}
\end{center}

This completes the proof that  $T$ and $T'$ are as described in the statement.
\end{proof}

We have proved so far that the minimum value of $D_0$ on $\TT_n$ is 3, and we have characterized those pairs  of trees $T,T'\in \TT_n$ such that $D_0(T,T')=3$. To extend this result to every $D_p$, $p\geq 1$, it is enough to check that every pair  of trees in $\TT_n$ such that $D_0(T,T')=3$ also satisfies that $D_p(T,T')=3$ for every $p\geq 1$, which is straightforward.
This completes the proof of Proposition  \ref{prop:mintt}.  
\medskip

\subsection*{Proof of Proposition \ref{prop:minbtt}}
As in Proposition \ref{prop:mintt}, 
we also split this proof  into several lemmas. 
First of all, notice that there are pairs of trees $T,T'\in \BT_n$ such that $D_p(T,T') =4$ for every $p\in \{0\}\cup[1,\infty[$: see, for instance, Fig.~\ref{fig:dist4}.  Therefore, the minimum value of $D_p$ on $\BT_n$ is at most 4.

\begin{center}
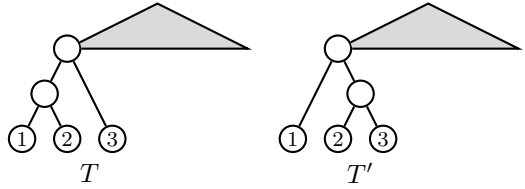
\begin{figure}[htb]
\begin{center}
\begin{tikzpicture}[thick,>=stealth,scale=0.3]
  \begin{scope}
\draw(0,0) node [tre] (1) {};  \etq 1
\draw(2,0) node [tre] (2) {};  \etq 2
\draw(4,0) node [tre] (3) {};  \etq 3 
\draw(1,2) node [tre] (a) {};   
\draw(2,4) node [tre] (x) {};  
\draw (x)--(a);
\draw (a)--(1);
\draw (a)--(2);
\draw (x)--(3);
\draw[fill=black!15] (x)--(10,4)--(6,6)--(x);
\draw(3,-1.5) node   {$T$};
\end{scope}
\begin{scope}[xshift=12cm]
\draw(0,0) node [tre] (1) {};  \etq 1
\draw(2,0) node [tre] (2) {};  \etq 2
\draw(4,0) node [tre] (3) {};  \etq 3 
\draw(3,2) node [tre] (a) {};   
\draw(2,4) node [tre] (x) {};  
\draw (x)--(a);
\draw (a)--(3);
\draw (a)--(2);
\draw (x)--(1);
\draw[fill=black!15] (x)--(10,4)--(6,6)--(x);
\draw(3,-1.5) node   {$T'$};
\end{scope}
\end{tikzpicture}
\end{center}
\caption{\label{fig:dist4} 
A pair of binary trees such that $D_p(T,T') =4$. The grey triangles represent the same tree.}
\end{figure}
\end{center}

Notice also that Lemma \ref{lem:mtt0} also applies in $\BT_n$, and therefore, 
if $T,T'\in \BT_n$ are such that $D_0(T,T')>0$, then there exist two taxa $i\neq j$ such that $\varphi_T(i,j)\neq \varphi_{T'}(i,j)$. And, of course, Lemma \ref{lem:mtt1} also applies in  $\BT_n$.

\begin{lemma}\label{lem:mtt2bin}
Let $T,T'\in \BT_n$ be such that $D_0(T,T')\leq 4$. If $\varphi_T(i,j)=\varphi_{T'}(i,j)+m$, for some $1\leq i< j\leq n$ and some $m\geq 1$, then $m=1$.
\end{lemma}

\begin{proof}
Assume that  $\varphi_T(i,j)=\varphi_{T'}(i,j)+m$ with $m\geq 2$, and let us reach a contradiction.

If $\delta_{T'}(i)=\delta_T(i)$, then $\delta_{T'}(i)>\varphi_T(i,j)=\fTp(i,j)+m$, and therefore there exist  leaves $x_1,\ldots,x_m$ such that $\fT(i,x_l)=\fTp(i,j)+l$, for $l=1,\ldots,m$. By Lemma \ref{lem:mtt1}, each such leaf $x_l$ adds at least $1$ to $D_0(T,T')$. Therefore $D_0(T,T')\geq 1+m$. Now, if moreover
$\delta_{T'}(j)=\delta_T(j)$, then there also exist leaves $y_1,\ldots,y_m$ such that $\fT(j,y_l)=\fTp(i,j)+l$, for $l=1,\ldots,m$, and each such leaf $y_l$ also adds  at least $1$ to $D_0(T,T')$, which entails $D_0(T,T')\geq 1+2m\geq 5$. So,  if $D_0(T,T')\leq 4$, it must happen that $\delta_{T'}(i)\neq \delta_T(i)$ or $\delta_{T'}(j)\neq \delta_T(j)$ (or both). Let assume that  $\delta_{T'}(j)\neq \delta_T(j)$. 

Now, $\varphi_T(i,j)=\varphi_{T'}(i,j)+m\geq m$, and therefore there exist leaves $z_1,\ldots,z_m$ such that
$\fT(i,z_l)=\fT(j,z_l)=\fT(i,j)-l$, for $l=1,\ldots,m$. If $\fT(i,k_l)=\fTp(i,k_l)$, then
$$
\fTp(i,k_l)=\fT(i,k_l)=\fT(i,j)-l=\fTp(i,j)+(m-l)\geq \fTp(i,j)
$$
and therefore, by Lemma \ref{lem:mtt1}, $\fTp(j,k_l)\neq \fT(j,k_l)$, and thus, 
 each such leaf $z_l$ adds  at least $1$ to $D_0(T,T')$, which entails $D_0(T,T')\geq 2+m$. Therefore, if 
$D_0(T,T')\leq 4$ and $m\geq 2$, it must happen $m=2$ and, moreover,
$\fT(a,b)=\fTp(a,b)$ for every $(a,b)\neq (i,j),(j,j), (i,z_1), (i,z_2),  (j,z_1), (j,z_2)$.

In particular, $\delta_T(i)=\delta_{T'}(i)$, which as we have seen implies that there are at least two leaves $x_1,x_2$ such that $i\prec [i,x_2]_{T'}\prec [i,x_1]_{T'}\prec [i,j]_{T'}$. Since
$$
\fTp(z_1,z_2)=\fT(z_1,z_2)=\fT(i,j)-2=\fTp(i,j)
$$
implies that (up to interchanging $z_1$ and $z_2$)
$i\prec [i,z_1]_{T'}\prec  [i,j]_{T'}$ and $j\prec [j,z_2]_{T'}\prec  [i,j]_{T'}$, we conclude that $\{x_1,x_2,z_1,z_2\}$  are at least 3 different leaves and hence they contribute at least 3 to $D_0(T,T')$, making $D_0(T,T')\geq 5$.
\end{proof}

\begin{lemma}\label{lem:mtt4-bis}
Let $T,T'\in \BT_n$ be such that $D_0(T,T')\leq 4$. If $\varphi_T(i,j)=\varphi_{T'}(i,j)+1$, for some $1\leq i< j\leq n$, then $\delta_{T'}(i), \delta_{T'}(j)\leq  \varphi_{T'}(i,j)+2$.
\end{lemma}

\begin{proof}
Let us assume that $\delta_{T'}(i)\geq  \varphi_{T'}(i,j)+3$, and let us reach a contradiction. The case
when $\delta_{T'}(j)\geq  \varphi_{T'}(i,j)+3$ is symmetrical.

Since $\varphi_T(i,j)=\varphi_{T'}(i,j)+1>0$, there exists some taxon $k_0$ such that $[i,k_0]_T$ is the parent of $[i,j]_T$. Let us distinguish several cases.

\begin{enumerate}[(a)]
\item Assume that $\fT(i,k_0)=\fTp(i,k_0)$. Then, $\fTp(i,k_0)=\fT(i,k_0)=\fT(i,j)-1=\fTp(i,j)$ implies that $[j,k_0]_{T'}\prec [i,j]_{T'}$ and thus $\fTp(j,k_0)>\fTp(i,j)=\fT(i,j)-1=\fT(j,k_0)$ and in particular, by the previous lemma  $\fTp(j,k_0)=\fT(j,k_0)+1=\fT(i,j)=\fTp(i,j)+1$.
Now, since $D_0(T,T')\leq 4$, by Lemma \ref{lem:mtt2} the number of leaves $a\neq i,j,k_0$ such that  $a\prec [i,j]_{T'}$ is at most 2. 

  If $\delta_{T'}(i)\geq  \varphi_{T'}(i,j)+3$, then there exist leaves $k_1,k_2$ such that
$\fTp(i,k_1)=\fTp(i,j)-1$ and $\fTp(i,k_2)=\fTp(i,j)-2$ and then
$\fT(x,y)=\fTp(x,y)$ for every $(x,y)\neq (i,j), (i,k_0),(j,k_0),(k_1,i),(k_1,j),(k_2,i),(k_2,j)$.
In particular, no  leaf other than $i,j,k_0,k_1,k_2$ descends from $[i,j]_{T'}$.
But then
$$
\begin{array}{l}
\fT(k_1,k_0)=\fTp(k_1,k_0)=\fTp(i,j)=\fT(i,j)-1,\quad \fT(k_2,k_0)=\fT(i,j)-1\\
\fT(k_1,k_2)=\fTp(k_1,k_2)=\fTp(i,j)+1=\fT(i,j)
\end{array}
$$ 
imply that, up to interchanging $k_1$ and $k_2$, $i\prec [i,k_1]_{T}\prec [i,j]_{T}$ and 
$j\prec [j,k_2]_{T}\prec [i,j]_{T}$, and then
$$
\delta_{T'}(j)=\delta_T(j)>\fT(i,j)+1=\fTp(i,j)+2
$$
implies the existence of at least another leaf $h$ such that $j\prec [j,h]_{T'}\prec [j,k_0]_{T'}\prec [i,j]_{T'}$, which, as we have mentioned, is impossible.
So, this case cannot happen.

\item Assume now that $\fT(j,k_0)=\fTp(j,k_0)$. By symmetry with the previous case, this implies that 
$\fTp(i,k_0)=\fTp(i,j)+1$,  $\fTp(i,k_0)=\fT(i,k_0)+1$ and that  the number of leaves $a\neq i,j,k_0$ such that  $a\prec [i,j]_{T'}$ is at most 2. Now we have three new subcases to discuss.

\begin{enumerate}
\item[(b.1)] If $\delta_{T'}(i)= \varphi_{T'}(i,j)+4$, 
so that there exist leaves $k_1,k_2\neq i$ such that $\fTp(i,k_0),\fTp(i,k_1),\fTp(i,k_2)> \fTp(i,j)$, 
and no leaf other that $i,j,k_0,k_1,k_2$ descends from $[i,j]_{T'}$.
Then
$\fT(x,y)=\fTp(x,y)$ for every $(x,y)\neq (i,j), (i,k_0),(j,k_0),(k_1,i),(k_1,j),(k_2,i),(k_2,j)$. But in this case
it must happen that $\delta_T(j)=\delta_{T'}(j)=\fTp(i,j)+1=\fT(i,j)$, which is impossible. So, this case cannot happen.

\item[(b.2)] If $\delta_{T'}(i)= \varphi_{T'}(i,j)+3$ and $\delta_{T'}(j)= \varphi_{T'}(i,j)+2$, so that there exist leaves $k_1,k_2$ such that $\fTp(j,k_1)=\fTp(i,j)+1$,
$ \fTp(i,k_2)=\fTp(i,j)+2$ and, recall, $\fTp(i,k_0)=\fTp(i,j)+1$, then
$\fT(x,y)=\fTp(x,y)$ for every $(x,y)\neq (i,j), (i,k_0),(j,k_0),(k_1,i),(k_1,j),(k_2,i),(k_2,j)$. 
But then
$$
\fT(k_1,k_0)=\fTp(k_1,k_0)=\fTp(i,j)=\fT(i,j)-1
$$
implies that $k_1\prec [i,j]_{T}$, and then
$$
\begin{array}{l}
\delta_T(j)=\delta_{T'}(j)=\fTp(i,j)+2=\fT(i,j)+1,\\
\delta_T(k_1)=\delta_{T'}(k_1)=\fTp(i,j)+2=\fT(i,j)+1
\end{array}
$$
imply that $j$ and $k_1$ are the only children of $[i,j]_{T}$, which is, of course, impossible.
So, this case cannot happen, either.

\item[(b.3)] If $\delta_{T'}(i)= \varphi_{T'}(i,j)+3$ and $\delta_{T'}(j)= \varphi_{T'}(i,j)+1$, then on the one hand there exists a leaf $k_1$ such that $\fTp(i,k_1)=\fTp(j,k_0)-1=\fTp(i,j)-2$ and, on the other hand, as we have seen in (b.1), $\delta_T(j)> \delta_{T'}(j)$. Then,
$\fT(x,y)=\fTp(x,y)$ for every $(x,y)\neq (i,j), (j,j), (i,k_0),(j,k_0),(k_1,i),(k_1,j)$, and in particular no leaf other than $i,j,k_0,k_1$ descends from $[i,j]_{T'}$.

Now,
$$
\fT(k_1,k_0)=\fTp(k_1,k_0)=\fTp(i,j)+1=\fT(i,j)
$$
implies that $k_1\not\prec  [i,j]_{T}$, and 
$$
\delta_T(i)=\delta_{T'}(i)=\fTp(i,j)+3=\fT(i,j)+2
$$
implies that there exists a leaf $h\neq k_0,k_1$ such that $i\prec [i,h]_T\prec [i,j]_T$ and hence
$$
\fTp(i,h)=\fT(i,h)>\fT(i,j)+1=\fTp(i,j)
$$
would entail that $h\prec [i,j]_{T'}$,  which is impossible. Thus, this case cannot happen, either.
\end{enumerate}

\item Assume finally that $\fT(i,k_0)\neq \fTp(i,k_0)$ and  $\fT(j,k_0)\neq \fTp(j,k_0)$. The contribution to $D_0$ of the pairs $(i,j),(i,k_0),(j,k_0)$ is at least 3, and therefore there can  only exist at most one other pair of leaves with different cophenetic value in $T$ and in $T'$. Since every $x\neq i,j$ such that $x\prec [i,j]_{T'}$ defines at least one such pair,  we conclude that 
if  $\delta_{T'}(i)\geq \varphi_{T'}(i,j)+3$, then, it must happen that $[i,k_0]_{T'}\prec [i,j]_{T'}$ and that there can only exist one leaf $k_1\neq k_0,i$  such that $[i,k_1]_{T'}\prec [i,j]_{T'}$, and then, moreover $[i,k_0]_{T'}\neq [i,k_1]_{T'}$. In this case, 
$\fT(x,y)=\fTp(x,y)$ for every $(x,y)\neq (i,j),(i,k_0),(j,k_0),(k_1,i),(k_1,j)$. But then, in particular,
$\delta_{T'}(j)=\fTp(i,j)+1$ and $\delta_T(j)=\delta_{T'}(j)$, which implies $\delta_T(i)=\fT(i,j)$, which is impossible
\end{enumerate}
This finishes the proof that, if $D_0(T,T')\leq 4$, then
$\delta_{T'}(i)\leq  \varphi_{T'}(i,j)+2$ and $ \delta_{T'}(j)\leq  \varphi_{T'}(i,j)+2$.
\end{proof}

\begin{lemma}\label{lem:sibl}
Let $T,T'\in \BT_n$ be such that $D_0(T,T')\leq 4$. If $\varphi_T(i,j)=\varphi_{T'}(i,j)+1$, for some $1\leq i< j\leq n$, then $i,j$ are sibling in $T$.
\end{lemma}

\begin{proof}
Let $k_0$ be any leaf such that $[i,k_0]_T=[j,k_0]_T$ is the parent of $[i,j]_T$ in $T$.  If $\fT(i,k_0)=\fTp(i,k_0)$, then
$\fTp(i,k_0)=\fT(i,k_0)=\fT(i,j)-1=\fTp(i,j)$
implies that $[j,k_0]_{T'}\prec [i,j]_{T'}$ and thus $\fTp(j,k_0)>\fTp(i,j)=\fT(i,j)-1=\fT(j,k_0)$. Therefore,
$|\fT(i,k_0)-\fTp(i,k_0)|+|\fT(j,k_0)-\fTp(j,k_0)|\geq 1$.

Assume now that $i,j$ are not sibling in $T$, and let $h$ be a leaf such that $[i,h]_T$ is a child of $[i,j]_T$.
If $\fT(i,h)\leq \fTp(i,h)$, then
$$
\delta_{T'}(i)\geq \fTp(i,h)+1\geq \fT(i,h)+1= \fT(i,j)+2=\fTp(i,j)+3
$$
which is impossible by the previous lemma. Therefore, $\fT(i,h)> \fTp(i,h)$, and by Lemma \ref{lem:mtt2bin},
$\fT(i,h)= \fTp(i,h)+1$.

In a similar way,
if $\delta_T(i)=\delta_{T'}(i)$, then
$$
\delta_{T'}(i)=\delta_T(i)\geq \fT(i,h)+1= \fT(i,j)+2=\fTp(i,j)+3
$$
which is again impossible by the previous lemma.  Therefore, $\delta_T(i)\neq \delta_{T'}(i)$, too. So,
$(i,j)$, $(i,k_0)$, $(j,k_0)$, $(i,i)$, and $(i,h)$ contribute at least 4 to $D_0(T,T')\leq 4$, which implies that $\fT(x,y)=\fTp(x,y)$ for every other pair of leaves $(x,y)$.
But then, 
$$
\begin{array}{l}
\fTp(j,h)=\fT(j,h)=\fT(i,j)=\fTp(i,j)+1\\
\fTp(i,h)=\fT(i,h)-1=\fT(i,j)=\fTp(i,j)+1
\end{array}
$$
which is impossible. Therefore, $i$ and $j$ are sibling in $T$.
\end{proof}

\begin{lemma}\label{lem:nosibl}
Let $T,T'\in \BT_n$ be such that $D_0(T,T')\leq 4$. If $\varphi_T(i,j)=\varphi_{T'}(i,j)+1$, for some $1\leq i< j\leq n$, then $i,j$ are not sibling in $T'$.
\end{lemma}

\begin{proof}
Assume that $i,j$ are sibling in $T'$, and recall that we already know that they are sibling in $T$.
Let $k_0$ be any leaf such that $[i,k_0]_T=[j,k_0]_T$ is the parent of $[i,j]_T$ in $T$. If $\fT(i,k_0)=\fTp(i,k_0)$, then
$$
\fTp(i,k_0)=\fT(i,k_0)=\fT(i,j)-1=\fTp(i,j)
$$
which is impossible if $i,j$ are sibling in $T'$. Thus, $\fT(i,k_0)\neq \fTp(i,k_0)$ and, by symmetry, 
$\fT(j,k_0)\neq \fTp(j,k_0)$. On the other hand, if $\delta_T(i)=\delta_{T'}(i)$, then
$$
\delta_T(i)=\delta_{T'}(i)=\fTp(i,j)+1=\fT(i,j)
$$
which is also impossible. Therefore, $\delta_T(i)\neq \delta_{T'}(i)$ and, by symmetry, $\delta_T(j)\neq \delta_{T'}(j)$. But, then, $D_0(T,T')\geq 5$.
\end{proof}

Summarizing what we know so far, we have proved that if $D_0(T,T')\leq 4$ and $\varphi_T(i,j)\neq \varphi_{T'}(i,j)$, then, up to interchanging $T$ and $T'$, $\varphi_T(i,j)= \varphi_{T'}(i,j)+1$,
$i,j$ are sibling in $T$, and then the subtree of $T'$ rooted at $[i,j]_{T'}$ is a triplet or a totally balanced quartet; cf. Fig. \ref{fig:both}. Next  two lemmas cover  these  two possibilities.

\begin{center}
\begin{figure}[htb]
\begin{center}
\begin{tikzpicture}[thick,>=stealth,scale=0.3]
  \begin{scope}
\draw(0,0) node [tre] (i) {};  \etq i
\draw(2,0) node [tre] (k) {};  \etq k
\draw(4,0) node [tre] (j) {};  \etq j
\draw(1,2) node [tre] (a) {};   
\draw(2,4) node [tre] (x) {};  
\draw (x)--(a);
\draw (a)--(i);
\draw (a)--(k);
\draw (x)--(j);
\end{scope}
\begin{scope}[xshift=7cm]
\draw(0,0) node [tre] (i) {};  \etq i
\draw(2,0) node [tre] (k) {};  \etq k
\draw(4,0) node [tre] (l) {};  \etq l
\draw(6,0) node [tre] (j) {};  \etq j
\draw(1,2) node [tre] (a) {};   
\draw(5,2) node [tre] (b) {};   
\draw(3,4) node [tre] (x) {};  
\draw (x)--(a);
\draw (x)--(b);
\draw (a)--(i);
\draw (a)--(k);
\draw (b)--(l);
\draw (b)--(j);
\end{scope}
\end{tikzpicture}
\end{center}
\caption{\label{fig:both} 
The only possibilities for the subtree of $T'$ rooted at $[i,j]_{T'}$ if $D_0(T,T')\leq 4$ and $\varphi_T(i,j)= \varphi_{T'}(i,j)+1$.}
\end{figure}
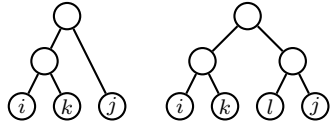
\end{center}

\begin{lemma}\label{lem:tri}
Let $T,T'\in \BT_n$ be such that $D_0(T,T')\leq 4$. If $\varphi_T(i,j)=\varphi_{T'}(i,j)+1$, for some $1\leq i< j\leq n$, and the subtree of $T'$ rooted at $[i,j]_{T'}$ is the triplet depicted in the left hand side of  Fig. \ref{fig:both}, then $T$ is obtained from $T'$ by interchanging $j$ and $k$: cf.\ Fig. \ref{fig:dist4tri}. And, then $D_0(T,T')= 4$.
\end{lemma}

\begin{center}
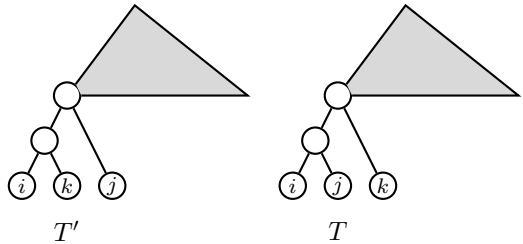
\begin{figure}[htb]
\begin{center}
\begin{tikzpicture}[thick,>=stealth,scale=0.3]
  \begin{scope}
\draw(0,0) node [tre] (i) {};  \etq i
\draw(2,0) node [tre] (k) {};  \etq k
\draw(4,0) node [tre] (j) {};  \etq j
\draw(1,2) node [tre] (a) {};   
\draw(2,4) node [tre] (x) {};  
\draw[fill=black!15] (x)--(10,4)--(5,8)--(x);
\draw (x)--(a);
\draw (a)--(i);
\draw (a)--(k);
\draw (x)--(j);
\draw (2,-2) node {$T'$};
\end{scope}
\begin{scope}[xshift=12cm]
\draw(0,0) node [tre] (i) {};  \etq i
\draw(2,0) node [tre] (j) {};  \etq j
\draw(4,0) node [tre] (k) {};  \etq k
\draw(1,2) node [tre] (a) {};   
\draw(2,4) node [tre] (x) {};  
\draw[fill=black!15] (x)--(10,4)--(5,8)--(x);
\draw (x)--(a);
\draw (a)--(i);
\draw (a)--(j);
\draw (x)--(k);
\draw (2,-2) node {$T$};
\end{scope}
\end{tikzpicture}
\end{center}
\caption{\label{fig:dist4tri} 
The only pairs of trees $T,T'$ such that $D_0(T,T')\leq 4$ and $\varphi_T(i,j)= \varphi_{T'}(i,j)+1$, when the subtree of $T'$ rooted at $[i,j]_{T'}$ is a triplet.}
\end{figure}
\end{center}

\begin{proof}
Assume that the subtree of $T'$ rooted at $[i,j]_{T'}$ has the form depicted in the left hand side of  Fig. \ref{fig:both}, and that $\varphi_T(i,j)=\varphi_{T'}(i,j)+1$. Then, since $i$ and $j$ are sibling in $T$,
$$
\delta_T(j)=\fT(i,j)+1=\fTp(i,j)+2=\delta_{T'}(j)+1.
$$
Now, if $\fT(i,k)\geq \fTp(i,k)$, then
$$
\fT(i,k)\geq \fTp(i,k)=\fTp(i,j)+1=\fT(i,j)
$$
which is impossible, because $i$ and $j$ are sibling in $T$. Therefore, $\fT(i,k)< \fTp(i,k)$ and, by Lemma \ref{lem:mtt2bin}, $\fT(i,k)=\fTp(i,k)-1$, and in particular
$\fT(i,k)=\fT(j,k)=\fT(i,j)-1$. Therefore, $[i,k]_T$ is the parent of $[i,j]_T$ in $T$.

Finally, if $\delta_T(k)\geq \fT(i,j)+1$, then there exists at least some other leaf $l\prec [i,k]_T=[j,k]_T$. But then $\fT(i,l)\neq \fTp(i,l)$, because otherwise
$$
\fTp(i,l)= \fT(i,l)=\fT(i,j)-1=\fTp(i,j),
$$
which is impossible because the only leaves descending from $[i,j]_{T'}$ are $i,j,k$. And, by symmetry $\fT(j,l)\neq \fTp(j,l)$, and we reach $D_0(T,T')\geq 5$.
Therefore,
$$
\delta_T(k)=\fT(i,j)=\fTp(i,j)+1=\delta_{T'}(k)-1.
$$
So, in summary, $\varphi_T(i,j)=\varphi_{T'}(i,j)+1$, $\delta_T(j)=\delta_{T'}(j)+1$, $\fT(i,k)=\fTp(i,k)-1$, 
and $\delta_T(k)=\delta_{T'}(k)-1$, and  
$\fT(x,y)=\fTp(x,y)$ for every $(x,y)$ other than $(i,j),(j,j),(i,k),(k,k)$. Moreover,  in $T$, $k$ is the other child of the parent of $[i,j]_T$.

So, the subtree $T_0$ of $T$ rooted at the parent of $[i,j]_T$ is obtained by interchanging $j$ and $k$  in the subtree $T_0'$ of $T'$ rooted at $[i,j]_{T'}$. Finally, let us prove now that $T$ and $T'$ are exactly the same except for $T_0$ and $T_0'$. 
More specifically, let $T_1$ and $T_1'$ be obtained by replacing in $T$ and $T'$ the subtrees $T_0$ and $T_0'$ by a single leaf $x$. Since for every $p,q\notin \{i,j,k\}$,
$$
\begin{array}{l}
\varphi_{T_1'}(p,q)=\fTp(p,q)=\fT(p,q)=\varphi_{T_1}(p,q),\\
\varphi_{T_1'}(x,p)=\fTp(i,p)=\fT(i,p)=\varphi_{T_1}(x,p),
\end{array}
$$
we deduce,  by Theorem \ref{th:char}, that $T_1=T_1'$. 

This completes the proof that  $T$ is obtained from $T'$ by interchanging the leaf $j$ and its nephew $k$. 
\end{proof}

\begin{lemma}\label{lem:four}
Let $T,T'\in \BT_n$ be such that $D_0(T,T')\leq 4$. If $\varphi_T(i,j)=\varphi_{T'}(i,j)+1$, for some $1\leq i< j\leq n$, and the subtree of $T'$ rooted at $[i,j]_{T'}$ is the quartet depicted in the right hand side of  Fig. \ref{fig:both}, then $T$ is obtained from $T'$ by interchanging $j$ and $k$: cf.\ Fig. \ref{fig:dist4tfour}. And, then $D_0(T,T')= 4$.
\end{lemma}

\begin{center}
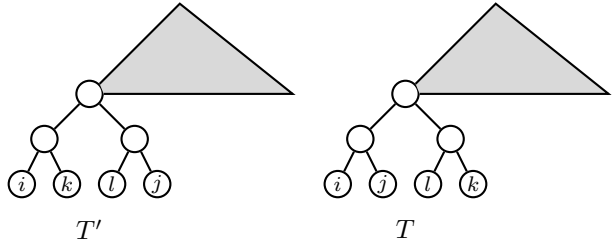
\begin{figure}[htb]
\begin{center}
\begin{tikzpicture}[thick,>=stealth,scale=0.3]
  \begin{scope}
\draw(0,0) node [tre] (i) {};  \etq i
\draw(2,0) node [tre] (k) {};  \etq k
\draw(4,0) node [tre] (l) {};  \etq l
\draw(6,0) node [tre] (j) {};  \etq j
\draw(1,2) node [tre] (a) {};   
\draw(5,2) node [tre] (b) {};   
\draw(3,4) node [tre] (x) {};  
\draw[fill=black!15] (x)--(12,4)--(7,8)--(x);
\draw (x)--(a);
\draw (x)--(b);
\draw (a)--(i);
\draw (a)--(k);
\draw (b)--(l);
\draw (b)--(j);
\draw (3,-2) node {$T'$};
\end{scope}
\begin{scope}[xshift=14cm]
\draw(0,0) node [tre] (i) {};  \etq i
\draw(2,0) node [tre] (j) {};  \etq j
\draw(4,0) node [tre] (l) {};  \etq l
\draw(6,0) node [tre] (k) {};  \etq k
\draw(1,2) node [tre] (a) {};   
\draw(5,2) node [tre] (b) {};   
\draw(3,4) node [tre] (x) {};  
\draw[fill=black!15] (x)--(12,4)--(7,8)--(x);
\draw (x)--(a);
\draw (x)--(b);
\draw (a)--(i);
\draw (a)--(j);
\draw (b)--(l);
\draw (b)--(k);
\draw (3,-2) node {$T$};
\end{scope}
\end{tikzpicture}
\end{center}
\caption{\label{fig:dist4tfour} 
The only pairs of trees $T,T'$ such that $D_0(T,T')\leq 4$ and $\varphi_T(i,j)= \varphi_{T'}(i,j)+1$, when the subtree of $T'$ rooted at $[i,j]_{T'}$ is a quartet.}
\end{figure}
\end{center}

\begin{proof}
Assume that the subtree of $T'$ rooted at $[i,j]_{T'}$ has the form depicted in the right hand side of  Fig. \ref{fig:both}, and that $\varphi_T(i,j)=\varphi_{T'}(i,j)+1$. 

If  $\fT(i,k)\geq \fTp(i,k)$, then
$$
\fT(i,k)\geq  \fTp(i,k)=\fTp(i,j)+1=\fT(i,j)
$$
which is impossible if $i,j$ are sibling in $T$. Therefore, $\fT(i,k)<  \fTp(i,k)$ and, by  Lemma \ref{lem:mtt2bin}, $\fT(i,k)=\fTp(i,k)-1$, and in particular $\fT(i,k)=\fT(i,j)-1$. By symmetry,
$\fT(j,l)=\fTp(j,l)-1$ and hence $\fT(j,l)=\fT(i,j)-1$, too. Therefore, both $k$ and $l$ are descendants of the parent of $[i,j]_T$. But then, 
$$
\fTp(k,l)=\fTp(i,j)=\fT(i,j)-1<\fT(k,l)
$$
and therefore, by  Lemma \ref{lem:mtt2bin}, $\fT(k,l)=\fTp(k,l)+1=\fT(i,j)$.

At this point,  $D_0(T,T')\leq 4$ entails that
$\fT(x,y)=\fTp(x,y)$ for every $(x,y)$ other than $(i,j),(i,k),(j,l),(k,l)$. Moreover, $i,k,j,l$ are the only descendant leaves of the parent of $[i,j]_T$ in $T$. Indeed, if $h$ is another descendant leaf of the parent of $[i,j]_{T'}$, then
$$
\fTp(i,h)=\fT(i,h)=\fT(i,j)-1=\fTp(i,j)
$$
and therefore $h$ would be another descendant of $[i,j]_{T'}$.
And, as we have seen,
the subtree $T_0$ of $T$ rooted at this node is obtained from the subtree $T_0'$ of $T'$ rooted at $[i,j]_{T'}$ by interchanging $j$ and $k$. Finally, arguing as in the last part of the proof of the previous lemma, we deduce that $T$ and $T'$ are exactly the same except for $T_0$ and $T_0'$.
\end{proof}

We have proved so far that the minimum value of $D_0$ on $\BTT_n$ is 4, and we have characterized the pairs  of trees $T,T'\in \BTT_n$ such that $D_0(T,T')=4$. To extend this result to every $D_p$, $p\geq 1$, it is enough to check that every pair of binary trees such that $D_0(T,T')=4$ also satisfies that $D_p(T,T')=4$ for every $p\geq 1$, which is straightforward.
This completes the proof of Proposition  \ref{prop:minbtt}.


\subsection*{Proof of Proposition \ref{prop:diam}}

Let  $X_n$ denote any space $\UTT_n$, $\TT_n$ or $\BTT_n$, and let $\Delta_p(X_n)$, $p\in \{0\}\cup[1,\infty[$, denote the diameter of $d_{\varphi,p}$ on $X_n$. 

\begin{figure}[htb]
\begin{center}
\begin{tikzpicture}[thick,>=stealth,scale=0.4]
  \begin{scope}
\draw(0,0) node [tre] (1) {};  \etq 1
\draw(2,0) node [tre] (2) {};  \etq 2
\draw(4,0) node [tre] (3) {};  \etq 3
\draw(6,0) node  {$\ldots$};  
\draw(8,0) node [tre] (n) {};  \etq n
\draw(4,3) node[tre] (r) {};
\draw  (r)--(1);
\draw  (r)--(2);
\draw  (r)--(3);
\draw  (r)--(n);
\draw(4,-2) node {(a)};
\end{scope}
\begin{scope}[xshift=10cm]
\draw(0,0) node [tre] (1) {};  \etq 1
\draw(2,0) node [tre] (2) {};  \etq 2
\draw(4,0) node [tre] (3) {};  \etq 3
\draw(6,0) node  [tre] (4) {};  \etq 4
\draw(8,0) node [tre] (5) {};  \etq 5
\draw(1,2) node[tre] (a) {};
\draw(7,2) node[tre] (b) {};
\draw(5.5,3.5) node[tre] (c) {};
\draw(4,5) node[tre] (r) {};
\draw (a)--(1);
\draw (a)--(2);
\draw (b)--(4);
\draw (b)--(5);
\draw (c)--(b);
\draw (c)--(3);
\draw (r)-- (a);
\draw (r)-- (c);
\draw(4,-2) node {(b)};
\end{scope}
\begin{scope}[xshift=20cm]
\draw(0,0) node [tre] (1) {};  \etq 1
\draw(2,0) node [tre] (2) {};  \etq 2
\draw(4,0) node [tre] (3) {};  \etq 3
\draw(6,0) node  {$\ldots$};  
\draw(8,0) node [tre] (n) {};  \etq n
\draw(1,2) node[tre] (a) {};
\draw(3,3) node[tre] (b) {};
\draw(7,5) node[tre] (r) {};
\draw(5,4) node  {.};
\draw(5.3,4.15) node  {.};
\draw(5.6,4.3) node  {.};
\draw (a)--(1);
\draw (a)--(2);
\draw (b)--(3);
\draw (b)--(a);
\draw (b)-- (4.5,3.75);
\draw (r)-- (6,4.5);
\draw (r)-- (n);
\draw(4,-2) node {(c)};
\end{scope}
\end{tikzpicture}
\quad
\end{center}
\caption{\label{fig:exs} 
(a) The rooted star  with $n$ leaves.  (b) The only maximally balanced tree with 5 leaves, up to relabelings. (c) A rooted caterpillar with $n$ leaves. }
\end{figure}
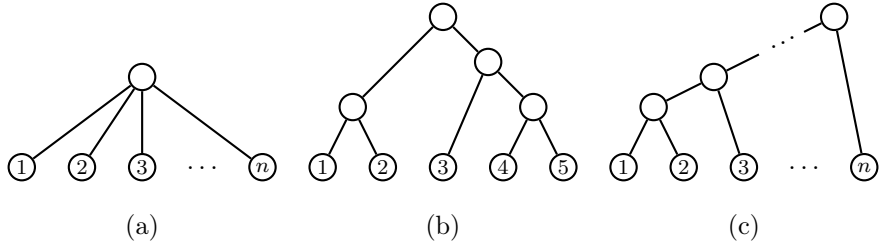

We consider first the case $p=1$, which will  be used later to prove the case $p>1$.
For every $T\in \UTT_n$, let
$$
S(T)=\sum_{i=1}^n \delta_T(i),\qquad \Phi(T)=\sum_{1\leq i<j\leq n} \varphi_T(i,j).
$$ 
$S$ and $\Phi$ are the extensions to $\UTT_n$ of the \emph{Sackin index} \cite{Sackin:72} and the \emph{total cophenetic index} \cite{MRR:12} for phylogenetic trees without nested taxa, respectively. Notice that
$\|\varphi(T)\|_1=S(T)+\Phi(T)$.
We have the following results on these indices:
\begin{itemize}
\item It is straightforward to check that the minimum values of $S(T)$ and $\Phi(T)$ on $\TT_n$ are both reached at the \emph{rooted star tree} with $n$ leaves (the phylogenetic tree with  all its leaves of depth 1; see Fig. \ref{fig:exs}.(a)), and these minimum values are, respectively,
$$
\min S(\TT_n)=n,\quad \min \Phi(\TT_n)=0.
$$ 
 
\item It is also straightforward  to check that the minimum values of $S(T)$ and $\Phi(T)$ on $\UTT_n$ are both reached at the rooted star tree with $n-1$ leaves and with the root labeled with  $n$, and  these minimum values are, respectively,
$$
\min S(\UTT_n)=n-1,\quad \min \Phi(\UTT_n)=0.
$$

\item The minimum values of $S(T)$ and $\Phi(T)$ on $\BTT_n$ are both reached at the \emph{maximally balanced trees}  with $n$ leaves (those binary trees such that, for every internal node, the numbers of descendant leaves of its two children differ at most in 1; see, for instance, Fig. \ref{fig:exs}.(b)).  And then, these minimum values are, respectively,
$$
\begin{array}{l}
\min S(\BTT_n)=    n\lfloor\log_2(4n)\rfloor-2^{\lfloor \log_2(2n)\rfloor}\\
 \min \Phi(\BTT_n)= \displaystyle\sum_{k=0}^{n-1} a(k),\mbox{ where $a(k)$ is the highest power of 2 that divides $n!$}
 \end{array}
$$ 
For the proofs, see \cite{Shao:90}  combined with \cite{Morris} for $S$, and \cite{MRR:12} for $\Phi$.  
From the first formula it is clear that $\min S(\BTT_n)$ is in $\Theta(n\log(n))$. As far as   
$ \min \Phi(\BTT_n)$ goes, it is shown in  \cite{MRR:12} that it satisfies the recurrence
$$
 \min \Phi(\BTT_n)= \min \Phi(\BTT_{\lceil n/2\rceil})+ \min \Phi(\BTT_{\lfloor n/2\rfloor})+\binom{\lceil n/2\rceil}{2}+\binom{\lfloor n/2\rfloor}{2},\quad \mbox{ for }n\geq 3
$$
from where it is obvious that its order is  in $\Theta(n^2)$.
 
\item The maximum values of $S(T)$ and $\Phi(T)$ on both $\TT_n$ and $\BTT_n$ are reached at the \emph{rooted caterpillar trees}  with $n$ leaves (binary phylogenetic trees  such that all their internal nodes have a leaf child; see Fig. \ref{fig:exs}.(c)). And then, these maximum values are, respectively,
$$
\max S(\TT_n)=\max S(\BTT_n)=\binom{n+1}{2}-1,\quad \max \Phi(\TT_n)=\max \Phi(\BTT_n)=\binom{n}{3},
$$ 
which are thus in $\Theta(n^2)$ and $\Theta(n^3)$, respectively. For the proofs, see again \cite{Shao:90}  for $S$ and \cite{MRR:12} for $\Phi$.

\item Given any tree in $\UTT_n$ with a nested taxon, if we replace this nested taxon by a new leaf labeled with it pending from the node previously labeled with it (cf.\ Fig.\ \ref{fig:unnesting}),  we obtain a new tree in $\UTT_n$ with strictly larger value of $S$ and the same value of $\Phi$. This shows that the maximum values of $S(T)$ and $\Phi(T)$ on $\UTT_n$ are reached at trees in $\TT_n$, and hence  at the {rooted caterpillar trees}  with $n$ leaves. Therefore, they are also 
in $\Theta(n^2)$ and $\Theta(n^3)$, respectively. 
\end{itemize}

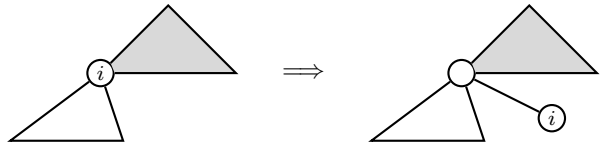
\begin{figure}[htb]
\begin{center}
\begin{tikzpicture}[thick,>=stealth,scale=0.3]
  \begin{scope}
\draw(7,3) node[tre] (r) {};  
\draw(r) node  {\footnotesize $i$};
\draw[fill=black!15] (r)--(13,3)--(10,6)--(r);
\draw (r)--(3,0)--(8,0)--(r);
\end{scope}
\begin{scope}[xshift=16cm]
\draw(0,0) node  {\ };  
\draw(0,6) node  {\ };  
\draw(0,3) node  {$\Longrightarrow$};
\end{scope}
\begin{scope}[xshift=16cm]
\draw(7,3) node[tre] (r) {};  
\draw(11,1) node[tre] (i) {}; \etq i
\draw[fill=black!15] (r)--(13,3)--(10,6)--(r);
\draw (r)--(3,0)--(8,0)--(r);
\draw (r)--(i);
\end{scope}
\end{tikzpicture}
\end{center}
\caption{\label{fig:unnesting} This operation increases the value of $S$ and does not modify the value of $\Phi$.}
\end{figure}

From these properties we deduce the following result.  
  
\begin{lemma}
The minimum value of  $\|\varphi(T)\|_1$ on $\UTT_n$ and $\TT_n$ is in $\Theta(n)$.
The minimum value of  $\|\varphi(T)\|_1$ on $\BTT_n$ is at most in $\Theta(n^2)$.
The maximum value of $\|\varphi(T)\|_1$ on $\UTT_n$, $\TT_n$ and $\BTT_n$ is in $\Theta(n^3)$. \qed
\end{lemma}

Now, we can apply this lemma to find the order of the diameter of $d_{\varphi,1}$ on the spaces $X_n$ of unweighted phylogenetic trees.

\begin{lemma} 
The diameter of $d_{\varphi,1}$  on $\UTT_n$, $\TT_n$ and $\BTT_n$ is in $\Theta(n^3)$.
\end{lemma}

\begin{proof}
Let 
$T_1,T_2\in X_n$. Then, on the one hand,
$$
d_{\varphi,1}(T_1,T_2) =\|\varphi(T_1)-\varphi(T_2)\|_1\leq \|\varphi(T_1)\|_1+\|\varphi(T_2)\|_1 \leq
2\cdot \max\|\varphi(X_n)\|_1 =\Theta(n^3)
$$
which shows that $\Delta_1(X_n)\leq O(n^3)$.
On the other hand, if $\|\varphi(T_1)\|_1\geq \|\varphi(T_2)\|_1$, then
$$
d_{\varphi,1}(T_1,T_2)=\|\varphi(T_1)-\varphi(T_2)\|_1\geq   \|\varphi(T_1)\|_1-\|\varphi(T_2)\|_1
$$
and therefore $\Delta_1(X_n)\geq \max\|\varphi(X_n)\|_1 - \min\|\varphi(X_n)\|_1$, which is again in $O(n^3)$.
This shows that $\Delta_1(X_n)$ is in $\Theta(n^3)$, as we claimed.
\end{proof}

Let us consider now the case $p>1$. Since, for every $x\in \RR^{m}$, $
\|x\|_1 \leq m^{1-\frac{1}{p}}\|x\|_p$,
we have that, for every pair of trees $T_1,T_2\in X_n$,
$$
d_{\varphi,1}(T_1,T_2) \leq  \binom{n+1}{2}^{1-\frac{1}{p}} d_{\varphi,p}(T_1,T_2).
$$
and therefore
$$
\Delta_1(X_n)\leq \binom{n+1}{2}^{1-\frac{1}{p}} \Delta_p (X_n),
$$
from where we deduce that 
$$
\Delta_p(X_n)\geq \Delta_1(X_n)\cdot \binom{n+1}{2}^{-1+\frac{1}{p}}=O(n^{(p+2)/p}).
$$

To prove the converse inequality, let
$$
\varphi^{(p)}(T)=\sum_{1\leq i\leq j \leq n} \varphi_T(i,j)^p.
$$
We have that, for every $T_1,T_2\in X_n$,
$$
\begin{array}{rl}
d_{\varphi,p}(T_1,T_2) &=\|\varphi(T_1)-\varphi(T_2)\|_p\leq \|\varphi(T_1)\|_p+\|\varphi(T_2)\|_p=
\sqrt[p]{\varphi^{(p)}(T_1)}+\sqrt[p]{\varphi^{(p)}(T_2)}\\ & \leq
2\sqrt[p]{\max \varphi^{(p)}(X_n)},
\end{array}
$$
which implies that $\Delta_p(X_n)\leq 2\sqrt[p]{\max \varphi^{(p)}(X_n)}$. 
Therefore, to prove that the diameter of $d_{\varphi,p}$ on each $X_n$ is bounded from above by
$O(n^{(p+2)/p})$, it is enough to prove that $\max \varphi^{(p)}(X_n)\leq O(n^{p+2})$. We do it in the next lemma.

\begin{lemma}
The maximum value of $\varphi^{(p)}(T)$ on $\UTT_n$, $\TT_n$ or $\BTT_n$ is reached at the rooted caterpillars, and its value is in $\Theta(n^{p+2})$.
\end{lemma}

\begin{proof}
Arguing as in the case $p=1$, we have that the maximum value of $\varphi^{(p)}(T)$ on $\UTT_n$ is reached on trees in $\TT_n$, because if we replace each nested taxon in a tree by a new  leaf labeled with the same taxon as in Fig. \ref{fig:unnesting}, the value of  $\varphi^{(p)}$ increases. On the other hand, if a tree $T\in \TT_n$ contains a node with $k\geq 3$ children, as in the left hand side of Fig.\ \ref{fig:bininc}, and we replace its subtree rooted at this node as described in the right hand side  of Fig.\ \ref{fig:bininc}, we obtain a new tree $T'\in \TT_n$ with larger $\varphi^{(p)}$ value: the values of $\varphi(i,j)^p$ for $i,j\in L(T_1)\cup\cdots\cup L(T_{k-1})$ increase, and the other values of  $\varphi(i,j)^p$ do not change. This implies that for every non-binary phylogenetic tree $T\in \TT_n$, there always exists a binary phylogenetic tree $T'\in\BTT_n$ such that $\varphi^{(p)}(T')>\varphi^{(p)}(T)$
and in particular  that the maximum value of $\varphi^{(p)}(T)$ on $\UTT_n$ is actually reached on  $\BTT_n$.

\begin{figure}[htb]
\begin{center}
\begin{tikzpicture}[thick,>=stealth,scale=0.25]
  \begin{scope}
\draw(0,0) node[tre] (z1) {}; 
\draw (z1)--(-2,-3)--(2,-3)--(z1);
\draw(0,-2) node  {\footnotesize $T_1$};
\draw(5,0) node[tre] (z2) {}; 
\draw (z2)--(3,-3)--(7,-3)--(z2);
\draw(5,-2) node  {\footnotesize $T_2$};
\draw(8,-2.8) node {.}; 
\draw(8.4,-2.8) node {.}; 
\draw(8.8,-2.8) node {.}; 
\draw(11.5,0) node[tre] (z3) {}; 
\draw (z3)--(9.5,-3)--(14.5,-3)--(z3);
\draw(12,-2) node  {\footnotesize $T_{k-1}$};
\draw(17.5,0) node[tre] (z4) {}; 
\draw (z4)--(15.5,-3)--(20.5,-3)--(z4);
\draw(17.5,-2) node  {\footnotesize $T_{k}$};
\draw(8.75,4) node[tre] (r) {};  
\draw[fill=black!15] (r)--(23,4)--(16,9)--(r);
\draw (r)--(z1);
\draw (r)--(z2);
\draw (r)--(z3);
\draw (r)--(z4);
\draw(8,-4) node  {$T$};
\end{scope}
\begin{scope}[xshift=26cm]
\draw(0,0) node[tre] (z1) {}; 
\draw (z1)--(-2,-3)--(2,-3)--(z1);
\draw(0,-2) node  {\footnotesize $T_1$};
\draw(5,0) node[tre] (z2) {}; 
\draw (z2)--(3,-3)--(7,-3)--(z2);
\draw(5,-2) node  {\footnotesize $T_2$};
\draw(8,-2.8) node {.}; 
\draw(8.4,-2.8) node {.}; 
\draw(8.8,-2.8) node {.}; 
\draw(11.5,0) node[tre] (z3) {}; 
\draw (z3)--(9.5,-3)--(14.5,-3)--(z3);
\draw(12,-2) node  {\footnotesize $T_{k-1}$};
\draw(17.5,0) node[tre] (z4) {}; 
\draw (z4)--(15.5,-3)--(20.5,-3)--(z4);
\draw(17.5,-2) node  {\footnotesize $T_{k}$};
\draw(6,2) node[tre] (x) {}; 
\draw(12,4) node[tre] (r) {};  
\draw[fill=black!15] (r)--(26,4)--(19,9)--(r);
\draw (x)--(z1);
\draw (x)--(z2);
\draw (x)--(z3);
\draw (r)--(x);
\draw (r)--(z4);
\draw(8,-3.9) node  {$T'$};
\end{scope}
\end{tikzpicture}
\end{center}
\caption{\label{fig:bininc} $\varphi^{(p)}(T')>\varphi^{(p)}(T)$.}
\end{figure}

\begin{figure}[htb]
\begin{center}
\begin{tikzpicture}[thick,>=stealth,scale=0.25]
  \begin{scope}
\draw(0,0) node[tre] (1) {};   \draw (1) node {\scriptsize $k$};  
\draw(2,2) node[tre] (a) {}; 
\draw(3,0) node [tre] (2) {}; \draw (2) node {\tiny $k\!\!\!-\!\!\!1$};  
\draw(4,4) node[tre] (b) {}; 
\draw(5,2) node [tre] (3) {}; \draw (3) node {\tiny $k\!\!\!-\!\!\!2$};  
\draw (b)--(5,5);
\draw (5.3,5.3) node {.};
\draw (5.6,5.6) node {.};
\draw (5.9,5.9) node {.};
\draw(7.5,7.5) node[tre] (c) {}; 
\draw (6.2,6.2)--(c);
\draw(8.5,4.5) node [tre] (k) {}; \draw (k) node {\scriptsize $1$};  
\draw(9.5,9.5) node[tre] (d) {}; 
\draw (d) node {\scriptsize $z$};  
\draw[fill=black!15] (d)--(20,9.5)--(15,14)--(d);
\draw(19,0) node[tre] (l) {};   \etq l
\draw(17,2) node[tre] (a1) {}; 
\draw(16,0) node [tre] (l-1) {}; 
\draw (l-1) node {\tiny $l\!\!\!-\!\!\!1$};  
\draw(15,4) node[tre] (b1) {}; 
\draw(14,2) node [tre] (l-2) {}; 
\draw (l-2) node {\tiny $l\!\!\!-\!\!\!2$};  
\draw (b1)--(14,5);
\draw (13.7,5.3) node {.};
\draw (13.4,5.6) node {.};
\draw (13.1,5.9) node {.};
\draw(11.5,7.5) node[tre] (c1) {}; 
\draw (12.8,6.2)--(c1);
\draw(10.5,4.5) node [tre] (k+1) {}; 
\draw (k+1) node {\tiny $k\!\!\!+\!\!\!1$};  
\draw  (a)--(1);
\draw (a)--(2);
\draw (b)--(a);
\draw  (b)--(3);
\draw (c)--(k);
\draw (d)--(c);
\draw  (a1)--(l);
\draw (a1)--(l-1);
\draw (b1)--(a1);
\draw  (b1)--(l-2);
\draw (c1)--(k+1);
\draw (d)--(c1);
\draw(9.5,-2) node  {$T$};
\end{scope}
\begin{scope}[xshift=26cm]
\draw(0,0) node[tre] (1) {};   \draw (1) node {\scriptsize $l$};  
\draw(2,2) node[tre] (a) {}; 
\draw(4,0) node[tre] (2) {};  \draw (2) node {\tiny $l\!\!\!-\!\!\!1$};  
\draw(4,3) node[tre] (b) {}; 
\draw(6,0) node[tre] (3) {}; \draw (3) node {\tiny $l\!\!\!-\!\!\!2$}; 
\draw (b)--(5,3.5);
\draw (5.3,3.65) node {.};
\draw (5.6,3.8) node {.};
\draw (5.9,3.95) node {.};
\draw(7.5,4.75) node[tre] (c) {}; 
\draw (6.2,4.1)--(c);
\draw(10,0) node[tre] (k) {}; \draw (k) node {\tiny $k\!\!\!+\!\!\!1$}; 
\draw(9.5,5.75) node[tre] (d) {}; 
\draw(12,0) node[tre] (k+1) {}; 
\draw (k+1) node {\scriptsize $k$};  
\draw(d)--(10,6); 
\draw (10.3,6.15) node {.};
\draw (10.6,6.3) node {.};
\draw (10.9,6.45) node {.};
\draw(12.5,7.25) node[tre] (e) {}; 
\draw (11.2,6.6)--(e);
\draw (e) node {\scriptsize $z$};  
\draw[fill=black!15] (e)--(22.5,7.25)--(17.5,11)--(e);
\draw(16,0) node[tre] (l) {};   
\draw (l) node {\tiny $1$};  
\draw  (a)--(1);
\draw (a)--(2);
\draw (b)--(a);
\draw  (b)--(3);
\draw (c)--(k);
\draw (d)--(c);
\draw (d)--(k+1);
\draw (e)--(l);
\draw(9.5,-2) node  {$T'$};
\end{scope}
\end{tikzpicture}
\end{center}
\caption{\label{fig:larg1} 
$\varphi^{(p)}(T')>\varphi^{(p)}(T)$.}
\end{figure}

Let now $T\in\BTT_n$ and assume that it is not a caterpillar. Therefore, it has an internal node $z$ of largest depth without any leaf child; in particular, all internal descendant nodes of $z$ have some leaf child. Thus, and up to a relabeling of its leaves, $T$ has the form represented in the left hand side  of Fig. \ref{fig:larg1}, for some $k\geq 2$ and some $l\geq k+2$.  Consider then the tree $T'$  depicted in  right hand side  of  Fig. \ref{fig:larg1},  where the grey triangle represents the same tree in both sides. It turns out that $\varphi^{(p)}(T')-\varphi^{(p)}(T)>0$. Indeed, if $q$ denotes the depth of the node $z$ in both trees, then
$$
\varphi_{T'}(i,j)^p-\varphi_{T}(i,j)^p=
\left\{
\begin{array}{ll}
(q+i)^p-(q+i+1)^p & \mbox{ if $1\leq i=j\leq k-1$}\\
0 & \mbox{ if $i=j=k$}\\
(q+i)^p-(q+i-k+1)^p & \mbox{ if $k+1\leq i=j\leq l-1$}\\
(q+l-1)^p-(q+l-k)^p & \mbox{ if $ i=j=l$}\\
(q+i-1)^p-(q+i)^p & \mbox{ if $1\leq i<j\leq k$}\\
(q+i-1)^p-(q+i-k)^p & \mbox{ if $k+1\leq i<j\leq l$}\\
(q+i-1)^p-q^p & \mbox{ if $1\leq i\leq k<j\leq l$}\\
0 & \mbox{ otherwise}
\end{array}
\right.
$$
Therefore,
$$
\begin{array}{rl}
\varphi^{(p)}(T')-\varphi^{(p)}(T) & \displaystyle= \sum_{i=1}^{k-1}\big((q+i)^p-(q+i+1)^p\big) + \sum_{i=k+1}^{l-1}\big((q+i)^p-(q+i-k+1)^p\big) \\
& \displaystyle\quad +(q+l-1)^p-(q+l-k)^p
 +\sum_{i=1}^{k-1}(k-i)\big((q+i-1)^p-(q+i)^p\big)\\
& \displaystyle\quad+
\sum_{i=k+1}^{l-1}(l-i)\big((q+i-1)^p-(q+i-k)^p\big)+
\sum_{i=1}^{k}(l-k)\big((q+i-1)^p-q^p\big)\\
 & \displaystyle= (q+1)^p-(q+k)^p + \sum_{i=1}^{l-k-1}\big((q+k+i)^p-(q+1+i)^p\big) \\
& \displaystyle\quad +(q+l-1)^p-(q+l-k)^p
 +\sum_{i=1}^{k-1}(k-i)\big((q+i-1)^p-(q+i)^p\big)\\
& \displaystyle\quad+
\sum_{i=1}^{l-k-1}(l-k-i)\big((q+k+i-1)^p-(q+i)^p\big)+
\sum_{i=1}^{k}(l-k)\big((q+i-1)^p-q^p\big)\\
\end{array}
$$
To prove that this sum is non-negative, let us write it as
$$
\varphi^{(p)}(T')-\varphi^{(p)}(T) = S_1 + S_2+S_3,
$$
where
$$
\begin{array}{rl}
S_1 & = \displaystyle\sum_{i=1}^{k-1}(k-i)\big((q+i-1)^p-(q+i)^p\big)+ \sum_{i=1}^{k}(l-k)\big((q+i-1)^p-q^p\big) \\
S_2 & =\displaystyle\sum_{i=1}^{l-k-1}\big((q+k+i)^p-(q+1+i)^p\big) +
\sum_{i=1}^{l-k-1}(l-k-i)\big((q+k+i-1)^p-(q+i)^p\big) \\
S_3 & = (q+1)^p-(q+k)^p+(q+l-1)^p-(q+l-k)^p
\end{array}
$$
Then
$$
\begin{array}{rl}
S_1 & = \displaystyle\sum_{i=1}^{k-1}(k-i)\big((q+i-1)^p-(q+i)^p\big)+ \sum_{i=1}^{k}(l-k)\big((q+i-1)^p-q^p\big), \\ 
& = \displaystyle\sum_{i=1}^{k-1} (k-i) (q+i-1)^p -\sum_{i=1}^{k-1} (k-i) (q+i)^p +\sum_{i=1}^{k}(l-k)\big((q+i-1)^p-q^p\big), \\
& =  \displaystyle\sum_{i=1}^{k-1} (k-i) (q+i-1)^p - \sum_{i=2}^{k} (k-i+1) (q+i-1)^p +(l-k)\sum_{i=1}^{k}(q+i-1)^p- k (l-k) q^p, \\
& = \displaystyle\sum_{i=1}^{k-1} (l-k-1) (q+i-1)^p + k q^p- (q+k-1)^p + (l-k) (q+k-1)^p-k (l-k) q^p, \\
& = \displaystyle  (l-k-1)\sum_{i=1}^{k}  \big((q+i-1)^p-q^p\big) >0\\[2ex]
S_2 & = \displaystyle\sum_{i=1}^{l-k-1}\big((q+k+i)^p-(q+1+i)^p\big) +
\sum_{i=1}^{l-k-1}(l-k-i)\big((q+k+i-1)^p-(q+i)^p\big) \\
& = \displaystyle \sum_{i=1}^{l-k-1} \big((q+k+i)^p-(q+1+i)^p\big) +
\sum_{i=0}^{l-k-1}(l-k-i-1)\big((q+k+i)^p-(q+i+1)^p\big) \\
& = \displaystyle \sum_{i=1}^{l-k-1} (l-k-i) \big((q+k+i)^p-(q+1+i)^p\big) +(l-k-1)\big( (q+k)^p - (q+1)^p\big)\\
& > \displaystyle (l-k-1)\big( (q+k)^p - (q+1)^p\big).

\end{array}
$$
and therefore
$$
\begin{array}{l}
\varphi^{(p)}(T')-\varphi^{(p)}(T)  = S_1+S_2+S_3\\
\qquad >  (l-k-1)\big( (q+k)^p - (q+1)^p\big) +(q+1)^p-(q+k)^p+(q+l-1)^p-(q+l-k)^p\\
\qquad = (l-k-2) \big( (q+k)^p -(q+1)^p\big) +(q+l-1)^p-(q+l-k)^p>0.
\end{array}
$$

This implies that no tree other than a rooted caterpillar can have the largest  $\varphi^{(p)}$ value in $\BTT_n$, and hence also in $\TT_n$ and $\UTT_n$.

Finally, if $K_n$ denotes the  rooted caterpillar with $n$ leaves in Fig. \ref{fig:exs}.(c),
$$
\varphi_{K_n}(i,j)^p=\left\{
\begin{array}{ll}
(n-1)^p & \mbox{ if $i= j=1$}\\
(n-i+1)^p & \mbox{ if $2\leq i= j\leq n$}\\
(n-j)^p & \mbox{ if $1\leq i< j\leq n$}\\
\end{array}\right.
$$
and thus
$$
\begin{array}{rl}
\varphi^{(p)}(K_n) & =
(n-2)\cdot 1^p+(n-3)\cdot 2^p+\cdots +2\cdot (n-3)^p+1\cdot (n-2)^p\\
& \qquad\quad+1^p+2^p+\cdots +(n-2)^p+(n-1)^p+(n-1)^p\\
& = (n-1)\cdot 1^p+(n-2)\cdot 2^p+\cdots +3\cdot (n-3)^p+2\cdot (n-2)^p+(n-1)^p+(n-1)^p\\
& = \displaystyle \sum_{k=1}^{n-1} (n-k)\cdot k^p+(n-1)^p
\end{array}
$$
Now, it turns out that
\begin{equation}\label{powersum}
\sum_{k=1}^{n-1} k^m=\frac{1}{m+1}n^{m+1}+O(n^m).
\end{equation}
This property is well known for natural numbers $m\in \NN$ \cite{MW}. For arbitrary real numbers $m>0$, it derives from the fact that
$$
\int_1^{n-1} (x-1)^m dx \leq  \sum_{k=1}^{n-1} k^m \leq \int_1^{n-1} x^m dx,
$$
and then
$$
\begin{array}{l}
\displaystyle \int_1^{n-1} (x-1)^mdx=\frac{1}{m+1}(n-2)^{m+1}=\frac{1}{m+1}n^{m+1}+O(n^m)
\\
\displaystyle \int_1^{n-1} x^mdx=\frac{1}{m+1}(n-1)^{m+1}=\frac{1}{m+1}n^{m+1}+O(n^m)
\end{array}
$$

So, by identity (\ref{powersum}), we have that
$$
 \sum_{k=1}^{n-1} (n-k)\cdot k^p+(n-1)^p= n\sum_{k=1}^{n-1}k^p- \sum_{k=1}^{n-1}k^{p+1}+O(n^p)
 =\Big(\frac{1}{p+1}-\frac{1}{p+2}\Big)n^{p+2}+O(n^{p+1})
 $$
and hence $\varphi^{(p)}(K_n) $ is in $\Theta(n^{p+2})$. 
\end{proof} 

Therefore, $O(n^{(p+2)/p})\leq \Delta_p(X_n)\leq O(n^{(p+2)/p})$, which shows that
the diameter of $d_{\varphi,p}$  on $\UTT_n$, $\TT_n$ and $\BTT_n$ is indeed in $\Theta(n^{(p+2)/p})$.

We finally prove the case $p=0$, which needs a completely different argument.
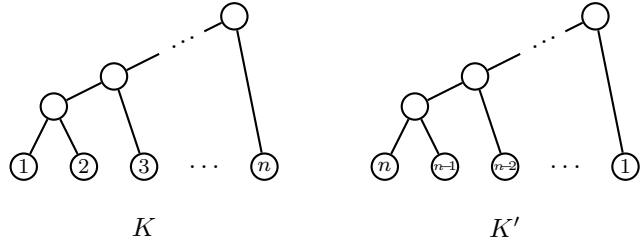
\begin{figure}[htb]
\begin{center}
\begin{tikzpicture}[thick,>=stealth,scale=0.4]
  \begin{scope}
\draw(0,0) node [tre] (1) {};  \etq 1
\draw(2,0) node [tre] (2) {};  \etq 2
\draw(4,0) node [tre] (3) {};  \etq 3
\draw(6,0) node  {$\ldots$};  
\draw(8,0) node [tre] (n) {};  \etq n
\draw(1,2) node[tre] (a) {};
\draw(3,3) node[tre] (b) {};
\draw(7,5) node[tre] (r) {};
\draw(5,4) node  {.};
\draw(5.3,4.15) node  {.};
\draw(5.6,4.3) node  {.};
\draw (a)--(1);
\draw (a)--(2);
\draw (b)--(3);
\draw (b)--(a);
\draw (b)-- (4.5,3.75);
\draw (r)-- (6,4.5);
\draw (r)-- (n);
\draw(4,-2) node {$K$};
\end{scope}
\begin{scope}[xshift=12cm]
\draw(0,0) node [tre] (1) {};   \draw (1) node {\footnotesize $n$}; 
\draw(2,0) node [tre] (2) {};  \draw (2) node {\tiny $n\!\!\!-\!\!\!1$}; 
\draw(4,0) node [tre] (3) {};  \draw (3) node {\tiny $n\!\!\!-\!\!\!2$}; 
\draw(6,0) node  {$\ldots$};  
\draw(8,0) node [tre] (n) {};   \draw (n) node {\footnotesize $1$}; 
\draw(1,2) node[tre] (a) {};
\draw(3,3) node[tre] (b) {};
\draw(7,5) node[tre] (r) {};
\draw(5,4) node  {.};
\draw(5.3,4.15) node  {.};
\draw(5.6,4.3) node  {.};
\draw (a)--(1);
\draw (a)--(2);
\draw (b)--(3);
\draw (b)--(a);
\draw (b)-- (4.5,3.75);
\draw (r)-- (6,4.5);
\draw (r)-- (n);
\draw(4,-2) node {$K'$};
\end{scope}
\end{tikzpicture}
\end{center}
\caption{\label{fig:twocat} 
The caterpillars used in the proof of Lemma \ref{lem:diam0}. }
\end{figure}

\begin{lemma}\label{lem:diam0}
The diameter of $d_{\varphi,0}$  on $\UTT_n$, $\TT_n$ and $\BTT_n$ is in $\Theta(n^2)$.
\end{lemma}

\begin{proof}
Since the cophenetic vector of a tree $T\in \UTT_n$ lies in $\RR^{n(n+1)/2}$, it is clear that
$d_{\varphi,0}(T_1,T_2)\leq n(n+1)/2$, for every $T_1,T_2\in \UTT_n$.
Now, consider the pair of rooted caterpillars with $n$ leaves depicted in Fig. \ref{fig:twocat}. We have that
$$
\begin{array}{lll}
\varphi_K(i,j)= n-j& \varphi_{K'}(i,j)=i-1  & \mbox{ for every $1\leq i<j\leq n$}\\
\varphi_K(i,i)= n-i+1& \varphi_{K'}(i,i)=i  & \mbox{ for every $2\leq i\leq n-1$}\\
\varphi_K(1,1)= n-1& \varphi_{K'}(1,1)=1  & \\
\varphi_K(n,n)= 1& \varphi_{K'}(n,n)=n-1  & 
\end{array}
$$
This shows that the number of pairs $(i,j)$, $1\leq i\leq j\leq n$, such that $\varphi_K(i,j)=\varphi_{K'}(i,j)$ is at most $(n+1)/2$, and therefore that $d_{\varphi,0}(K,K')$ is at least  $(n^2-1)/2$. 
So, the diameter of $d_{\varphi,0}$  on $ \UTT_n$ is bounded from above by $O(n^2)$, and its diameter on $\BTT_n$ is bounded from below by  $O(n^2)$, which implies that the 
diameter of $d_{\varphi,0}$  on $\UTT_n$, $\TT_n$ and $\BTT_n$ is in $\Theta(n^2)$.
\end{proof}

\end{bmcformat}
\end{document}